\documentclass[]{interact}

\usepackage{nomencl} 
\makenomenclature 

\usepackage[english]{babel}
\usepackage[usenames,dvipsnames,table,pdftex]{xcolor}
\usepackage{amsmath}
\usepackage{amssymb}
\usepackage{amsthm}
\usepackage{amsmath}
\usepackage{graphicx}
\usepackage{url}
\usepackage{float}
\usepackage{tikz}
\usepackage{pgfplots} 				
\usepackage[utf8]{inputenc}
\usepackage[hidelinks]{hyperref}
\usepackage{dsfont}   				
\usepackage[siunitx]{circuitikz} 

\usepackage{xargs} 				
\newcommandx{\F}[2][1=1, 2=2]{\frac{#1}{#2}}    

\usepackage{mathtools}
\usepackage[mathscr]{euscript}

\numberwithin{equation}{section}

\newcommand{\w}{\omega}

\newcommand{\p}{\partial}

\newcommand{\pH}{port-Hamiltonian}
\newcommand{\G}{\mathcal{G} }
\newcommand{\V}{\mathcal{V} }

\newcommand{\E}{\mathcal{E} }

\usepackage[prependcaption,colorinlistoftodos]{todonotes}

\definecolor{new}{rgb}{0.55,0,0.55}

\DeclareMathOperator{\col}{col}

\DeclareMathOperator{\diag}{diag}

\theoremstyle{plain}
\newtheorem{theorem}{Theorem}[section]

\newtheorem{proposition}[theorem]{Proposition}

\theoremstyle{definition}
\newtheorem{definition}[theorem]{Definition}

\newtheorem{assumption}[theorem]{Assumption} 

\theoremstyle{remark}
\newtheorem{remark}[theorem]{Remark}

\usepackage{epstopdf}
\usepackage[caption=false]{subfig}

\usepackage[numbers,sort&compress]{natbib}
\bibpunct[, ]{[}{]}{,}{n}{,}{,}
\renewcommand\bibfont{\fontsize{10}{12}\selectfont}

\begin{document}




\title{An energy-based analysis of reduced-order models of (networked) synchronous machines
}

\author{
\name{T.W. Stegink\textsuperscript{a}$^{\dagger}$\thanks{$^\dagger$Email of the corresponding authors: \{t.w.stegink, c.de.persis, a.j.van.der.schaft\}@rug.nl}
and C. {D}e Persis\textsuperscript{a}$^\dagger$ and A.J. van der Schaft\textsuperscript{b}$^\dagger$}
\affil{\textsuperscript{a}Engineering and Technology institute Groningen, University of Groningen, \\ Nijenborgh 4, 9747 AG Groningen, the Netherlands;
\\\textsuperscript{b}Johann Bernoulli Institute for Mathematics and Computer Science, University of Groningen, Nijenborgh 9, 9747 AG Groningen,
the Netherlands.}
}

\maketitle



\begin{abstract}
 Stability of power networks is an increasingly important topic because of the high penetration of renewable distributed generation units. This requires the development of advanced (typically model-based) techniques for the analysis and controller design of power networks. 
Although there are widely accepted reduced-order models to describe the dynamic behavior of power networks, they are commonly presented without details about the reduction procedure, hampering the understanding of the physical phenomena behind them. The present paper aims to provide a modular model derivation of multi-machine power networks. Starting from first-principle fundamental physics, we present detailed dynamical models of synchronous machines and clearly state the underlying assumptions which lead to some of the standard reduced-order multi-machine models, including the classical second-order swing equations.
In addition, the energy functions for the reduced-order multi-machine models are derived, which allows 
to represent the multi-machine systems as \pH\ systems. Moreover, the systems are proven to be passive with respect to its steady states, which permits for a power-preserving interconnection with other passive components, including passive controllers. As a result, the corresponding energy function or Hamiltonian can be used to provide a rigorous stability analysis of advanced  models for the power network without having to linearize 
the system. 

\end{abstract}



\section{Introduction}
\label{sec:introduction}


\subsection{Problem statement/motivation}

The control and stability of power networks has become increasingly challenging over the last decades. As renewable energy sources penetrate the grid, the conventional power plants have more difficulty in keeping the frequency around the nominal value, e.g. 50 Hz, leading to an increased chance of network failures or, in the worst case, even blackouts.  

The current developments require a sophisticated stability analysis of more advanced models for the power network as the grid is operating more frequently near its capacity constraints. For example, using high-order models of synchronous machines that better approximate the actual system allows us to establish results on the stability of power networks that are more reliable and accurate. 

However, in much of the recent literature, a rigorous stability analysis has been carried out only for low-order models of the power network which have a limited accuracy. For models of intermediate complexity the stability analysis has merely been done for the linearized system. Hence, a novel approach is required to make a profound stability analysis of these more complicated models possible. 

In this paper, we propose a unifying energy-based approach for the modelling and analysis of multi-machine power networks which is based on the theory of \pH\ systems. Since energy is the main quantity of interest, the \pH\ framework is a natural approach to deal with the problem \cite{phsurvey}. Moreover, it lends itself to deal with large-scale nonlinear multi-physics systems like power networks \cite{EJCFiaz,stegink2015port,LHMNLC,stegink2017unifying}. 

\subsection{Literature review}

The emphasis in the present paper lies on the modelling and analysis of (networked) synchronous machines since they have a crucial role in the stability of power networks as they are the most flexible and have to compensate for the increased fluctuation of both the supply and demand of power. 
An advanced model of the synchronous machine is the first-principle model which is derived in many power-engineering books \cite{anderson1977,powsysdynwiley,kundur}, see in particular  \cite[Chapter 11]{powsysdynwiley} for a detailed derivation of the  model. 

Modelling the first-principle synchronous (multi-)machine model using the theory of \pH\ systems has been done previously in \cite{EJCFiaz}. However, in this work, stabilization of the synchronous machine to the synchronous frequency could not be proven. 
In \cite{cal-tab}  a similar model for the synchronous machine is used, but with the damper windings neglected. Under some additional assumptions, asymptotic stability of a single machine is proven using a shifted energy function. For multi-machine systems, however, stability could not be proven using a similar approach.  

Summarizing, the complexity of the full-order model of the synchronous machine makes a rigorous stability analysis troublesome, especially when considering multi-machine networks, see also \cite{ortega2005transient}. Moreover, it is often not necessary to consider the full-order model when studying a particular aspect of the electromechanical dynamics (e.g. operation around the synchronous frequency) \cite{powsysdynwiley}.

On the other side of the spectrum, much of the literature using Lyapunov stability techniques rely on the second-order (non)linear swing equations as the model for the power network \cite{AGC_ACC2014,you2014reverse,zhangpapaautomatica,zhao2015distributedAC,pai1989energy,fouad1981transient,michel1983power,powsysdynwiley} or the third-order model as e.g. in \cite{trip2016internal}. For microgrids similar models are considered in which a Lyapunov stability analysis is carried out \cite{de2018bregman,de2016lyapunov}.  However, the models are often presented without stating the details on the model reduction procedure or the validity of the model. For example, the swing equations are inaccurate and only valid on a specific time scale up to the order of a few seconds so that asymptotic stability results have a limited value 
for the actual system  \cite{caliskan2015uses,anderson1977,kundur,powsysdynwiley}. 

Hence, it is appropriate to make simplifying assumptions for the full-order model and to focus on multi-machine models with intermediate complexity which provide a more accurate description of the network compared to the second- and third-order models \cite{anderson1977,kundur,powsysdynwiley}. In doing so, we explain how these intermediate-order models are obtained from the first-principle model and what the underlying assumption are. 
Here we follow the lines of \cite{powsysdynwiley}, where a detailed derivation of the reduced-order models is given.

\subsection{Contributions}
In the present literature the stability analysis of intermediate-order multi-machine models is only carried out for the linearized system  \cite{alv_meng_power_coupl_market,kundur,powsysdynwiley,anderson1977}. Consequently, the stability results are only valid around a specific operating point.
In particular, in \cite{alv_meng_power_coupl_market} a fourth-order model for the synchronous machine is considered which is coupled with market dynamics and the stability is analyzed by examining the eigenvalues of the linearized system. 

Our approach is different as the nonlinear nature of the power network is preserved. In particular, in this paper we consider, among other things, a nonlinear sixth-order reduced model of the synchronous machine that enables a quite accurate description of the power network while still allowing us to perform a rigorous stability analysis. 

In fact, in our previous work \cite{stegink2016optimal} we analyzed the sixth-order multi-machine model and we applied an optimal power dispatch controller and showed convergence using a suitable energy function. In the present work we will show that this energy function indeed corresponds to the electrical energy stored in the generator circuits and the transmission lines.

In addition, this paper  establishes a unifying energy-based analysis of intermediate-order models of (networked) synchronous machines for inductive networks. 
To this end, we provide a systematic way in obtaining the energy function of each reduced-order multi-machine model. Furthermore, it is shown that (a shifted variant of) these energy functions act as candidate Lyapunov functions for the stability analysis of power networks. 


In this respect, we show that the \pH\ framework is very convenient for representing the dynamics of the reduced-order multi-machine models and for the stability analysis. In particular it is shown that, using the physical energy stored in the synchronous machines and the transmission lines as the Hamiltonian, a \pH\ representation of the multi-machine power network is obtained. More specifically, while the system dynamics is complex, the interconnection and damping structure of the corresponding \pH\ system is sparse and, importantly, state-independent. The latter property implies shifted passivity of the intermediate-order models with respect to their steady states. This is property proves to be very convenient for control purposes \cite{phsurvey,stegink2017unifying,trip2016internal,stegink2016optimal}.

\subsection{Outline}

The remainder of the paper is structured as follows. First we state the preliminaries in Section \ref{sec:preliminaries}. Then in Section \ref{sec:full-order-model} the full-order first-principle model is presented and its \pH\ form is given. The model reduction procedure is discussed in Section \ref{sec:model-reduct-synchr} in which models of intermediate order are obtained. In Section \ref{sec:multi-machine-models} these models are used to establish multi-machine models, including the classical second-order model. Then in Section \ref{sec:energy-analysis} energy functions of the reduced order models are derived, which in Section \ref{sec:port-hamilt-repr} are used to put the multi-machine models in \pH\ form. Finally, Section \ref{sec:conclusions} discusses the conclusions and possible directions for future research. 


\section{Preliminaries}
\label{sec:preliminaries}
\subsection{Notation}
\label{sec:notation}
The set of real numbers and the set of complex numbers are respectively defined by $\mathbb R,\mathbb C$. Given a complex number $\alpha\in\mathbb C$, the real and imaginary part of are denoted by $\Re(\alpha),\Im(\alpha)$ respectively.  The imaginary unit is denoted by $j=\sqrt{-1}$. Let $\{v_1,v_2,\ldots,v_n\}$ be a set of real numbers, then $\diag(v_1,v_2,\ldots,v_n)$ denotes the $n\times n$ diagonal matrix with the entries $v_1,v_2,\ldots,v_n$ on the diagonal and likewise $\col(v_1,v_2,\ldots,v_n)$ denotes the column vector with the entries $v_1,v_2,\ldots,v_n$. Let $f:\mathbb R^n\to\mathbb R$ be a twice differentiable function, then $\nabla f(x)$ denotes the gradient of $f$ evaluated at $x$ and $\nabla^2 f(x)$ denotes the Hessian of $f$ evaluated at $x$. Given a symmetric matrix $A\in\mathbb R^{n\times n}$, we write $A>0 \ (A\ge0)$ to indicate that $A$ is a positive (semi-)definite matrix.

\subsubsection{Power network}
Consider a power grid consisting of $n$ buses. The network is represented by a connected and undirected graph $ \G = (\V, \E) $, where the set of nodes, $ \V = \{1,\ldots, n\}$, is the set of buses representing the synchronous machines and the set of edges, $ \E \subset \V \times \V  $, is the set of transmission lines connecting the buses where each edge $(i,k)=(k,i)\in \E$ is an unordered pair of two vertices $i,k\in\V$. Given a node $i$, then the set of neigboring nodes is denoted by $\mathcal N_i:=\{k \ | \ (i,k)\in\E\}$. 


\printnomenclature 

\nomenclature{$p$}{Angular momentum.}
\nomenclature{$\w_s$}{Synchronous frequency (e.g. 50Hz).}
\nomenclature{$\w$}{Rotor frequency.}
\nomenclature{$\Delta\w$}{Rotor frequency with respect to the synchronous rotating reference frame.}
\nomenclature{$T_{dq0}$}{The $dq0$-transformation or Park transformation.}
\nomenclature{$E_{q}'$}{$q$-axis component of the transient internal emf.}
\nomenclature{$E_{d}'$}{$d$-axis component of the transient internal emf.}
\nomenclature{$E_{q}''$}{$q$-axis component of the subtransient internal emf.}
\nomenclature{$E_{d}''$}{$d$-axis component of the subtransient internal emf.}
\nomenclature{$V_d$}{$d$-axis component of the external emf (of the synchronous machine).}
\nomenclature{$V_q$}{$q$-axis component of the external emf (of the synchronous machine).}
\nomenclature{$B_{ii}$}{Self-susceptance at node $i$.}
\nomenclature{$B_{ij}$}{Line susceptance between node  $i$ and $j$.}
\nomenclature{$X_{di}$}{$d$-axis synchronous reactance.}
\nomenclature{$X_{qi}$}{$q$-axis synchronous reactance.}
\nomenclature{$X_{di}'$}{$d$-axis  transient reactance.}
\nomenclature{$X_{qi}'$}{$q$-axis  transient reactance.}
\nomenclature{$X_{di}''$}{$d$-axis subtransient reactance.}
\nomenclature{$X_{qi}''$}{$q$-axis subtransient reactance.}
\nomenclature{$T_{doi}'$}{$d$-axis open-circuit transient time constant.}
\nomenclature{$T_{qoi}'$}{$q$-axis open-circuit transient time constant.}
\nomenclature{$T_{doi}''$}{$d$-axis open-circuit subtransient time constant.}
\nomenclature{$T_{qoi}''$}{$q$-axis open-circuit subtransient time constant.}
\nomenclature{$\Psi_d$}{$d$-axis stator winding flux linkage.}
\nomenclature{$\Psi_q$}{$q$-axis stator winding flux linkage.}
\nomenclature{$\Psi_0$}{$0$-axis stator winding flux linkage.}
\nomenclature{$\Psi_f$}{Field winding flux linkage.}
\nomenclature{$\Psi_g$}{Additional $q$-axis damper winding flux linkage.}
\nomenclature{$\Psi_D$}{$d$-axis damper winding flux linkage.}
\nomenclature{$\Psi_Q$}{$q$-axis damper winding flux linkage.}
\nomenclature{$d$}{Mechanical damping constant.}
\nomenclature{$D$}{Asynchronous damping constant.}
\nomenclature{$I_d$}{$d$-axis stator winding current.}
\nomenclature{$I_q$}{$q$-axis stator winding current.}
\nomenclature{$I_0$}{$0$-axis stator winding current.}
\nomenclature{$I_f$}{Field winding current.}
\nomenclature{$I_g$}{Additional $q$-axis damper winding current.}
\nomenclature{$I_D$}{$d$-axis damper winding current.}
\nomenclature{$I_Q$}{$q$-axis damper winding current.}
\nomenclature{$R$}{Stator winding resistance.}
\nomenclature{$\mathcal{B}$}{Incidence matrix of the network.}
\nomenclature{$\delta$}{Rotor angle with respect to the synchronous rotating reference frame.}
\nomenclature{$\gamma$}{Rotor angle of the synchronous machine.}
\nomenclature{$\theta$}{Voltage angle with respect to the synchronous rotating reference frame.}
\nomenclature{$\alpha$}{Voltage angle with respect to $dq0$-reference frame of the rotor.}


\subsection{The $dq0$-transformation}
An important coordinate transformation used in the literature on power systems is the $dq0$-transformation \cite{powsysdynwiley,EJCFiaz} or \emph{Park transformation} \cite{park1929two} which is defined by  
\begin{align}\label{eq:dq0tran}
  T_{dq0}(\gamma)&=\sqrt{\frac32}
\begin{bmatrix}
\cos (\gamma)  & \cos (\gamma-\frac{2\pi}{3}) & \cos (\gamma+\frac{2\pi}{3}) \\
\sin (\gamma)  & \sin (\gamma-\frac{2\pi}{3}) & \sin (\gamma+\frac{2\pi}{3}) \\
\frac1{\sqrt{2}} & \frac1{\sqrt{2}}               & \frac1{\sqrt{2}}
\end{bmatrix}.
\end{align}
Observe that the mapping \eqref{eq:dq0tran} is orthogonal, i.e., $T_{dq0}^{-1}(\gamma)= T_{dq0}^T(\gamma)$. The $dq0$-transformation offers various advantages when analyzing power system dynamics and is therefore widely used in applications. In particular, the $dq0$-transformation maps \emph{symmetric} or \emph{balanced} three-phase AC signals (see \cite[Section 2]{schiffer2016survey} for the definition) to constant signals. 
This significantly simplifies the modelling and analysis of power systems, which is the main reason why the transformation \eqref{eq:dq0tran} is used in the present case. In addition, the transformation \eqref{eq:dq0tran} exploits the fact that, in a power system operated under symmetric conditions, a three-phase signal can be represented by two quantities \cite{schiffer2016survey}.  

For example, for a synchronous machine with  AC voltage $V^{\text{ABC}}=\col(V^A,V^B,V^C)$ in the static ABC-reference frame, see Figure \ref{fig:SG}, the $dq0$-transformation is used to map this AC voltage to the (local) $dq0$-coordinates as $V^{dq0}=\col(V_d,V_q,V_0)=T_{dq0}(\gamma)V^{\text{ABC}}$. Note that the local $dq0$-reference is aligned with the rotor of the machine which has angle $\gamma$ with respect to the static ABC-reference frame, see again Figure \ref{fig:SG}. In case more that one synchronous machine is considered, then the voltage $V^{dq0^k}$ in local $dq0$-coordinates of machine $k$ can be expressed in the local $dq0$-coordinates of machine $i$ as
\begin{align}
V^{dq0^i}=T_{dq0}(\gamma_i)V^{ABC^i}=T_{dq0}(\gamma_i)V^{ABC^k}=T_{dq0}(\gamma_i)T_{dq0}(\gamma_k)^TV^{dq0^k}.\label{eq:dq0mappingTdq0}
\end{align}
An analogous expression can be obtained for relation between the currents $I^{dq0^i},$ and $I^{dq0^k}$. Here we can verify that
\begin{align*}
  T_{dq0}(\gamma_i)T_{dq0}(\gamma_k)^T=
  \begin{bmatrix}
    \cos \gamma_{ik} & -\sin \gamma_{ik}&0\\
    \sin \gamma_{ik} & \cos \gamma_{ik}&0\\
    0&0&1
  \end{bmatrix}
\end{align*}
where $\gamma_{ik}:=\gamma_i-\gamma_k$ represents the rotor angle difference between  synchronous machines $i$ and $k$ respectively. 

\subsection{Phasor notation}
When considering operation around the synchronous frequency, the voltages and currents can be represented as phasors in the $dq$-coordinates rotating at the synchronous frequency. We use the following notation for the phasor\footnote{This is in contrast to \cite{kundur,sauerpai1998powersystem} where the convention $\overline V=V_d+jV_q$ is used.} \cite{powsysdynwiley}:
\begin{alignat*}{3}
  \overline V&=\sqrt{V_q^2+V_d^2}\exp\Big({j\arctan{(\frac{V_d}{V_q})}}\Big)&&=\overline V_q+\overline V_d&&=V_q+jV_d,\\ \overline I&=\sqrt{I_q^2+I_d^2}\exp\Big({j\arctan{(\frac{I_d}{I_q})}}\Big)&&=\overline I_q+\overline I_d&&=I_q+jI_d,
\end{alignat*}
which is commonly used in the power system literature \cite{powsysdynwiley,schiffer2016survey}. Here the bar-notation is used to represent the complex phasor and we define $\overline V_q=V_q, \overline V_d=jV_d$ and a likewise $\overline I_q=I_q, \overline I_d=jI_d$ for the currents.  In this case, the mapping between the voltages (and current) from one $dq$-reference frame to another is given by 
\begin{align}\label{eq:phasor-trans}
  \begin{aligned}
  \overline V^{dq^i}&=e^{-j\gamma_{ik}}\overline V^{dq^k}=(\cos \gamma_{ik}-j \sin \gamma_{ik})(V_{q}^{dq^k}+jV_{d}^{dq^k})\\
&=V_{q}^{dq^k}\cos \gamma_{ik}+V_{d}^{dq^k}\sin \gamma_{ik}+j(V_{d}^{dq^k}\cos\gamma_{ik}-V_{q}^{dq^k}\sin \gamma_{ik}).
\end{aligned}
\end{align}
By equating the real and imaginary parts, this exactly corresponds to the transformation \eqref{eq:dq0mappingTdq0} as expected.   

\section{Full-order model of the synchronous machine}
\label{sec:full-order-model}

\begin{figure}
  \centering
  \includegraphics[width=0.4\linewidth]{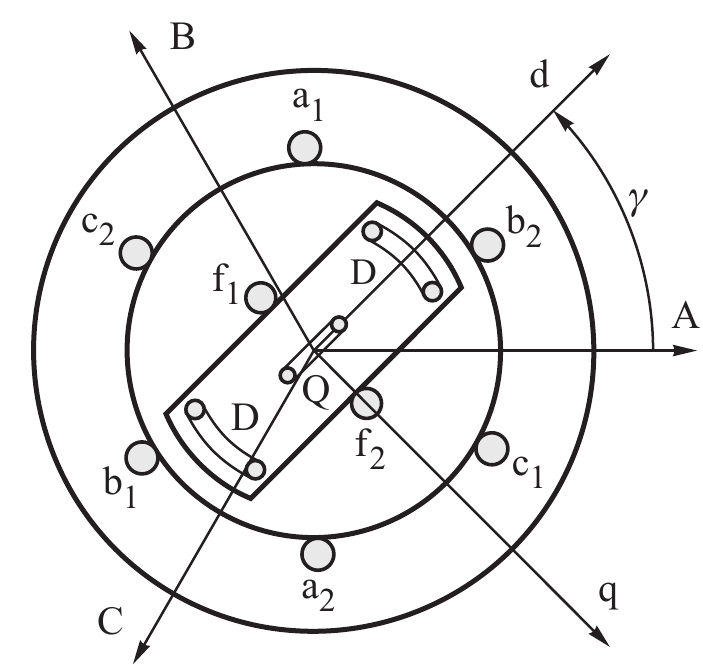}
  \caption{Schematic illustration of a  (salient-pole) synchronous machine \cite{powsysdynwiley}. }
  \label{fig:SG}
\end{figure}
A synchronous machine is a multi-physics system characterized by both mechanical and electrical variables, i.e., an electromechanical system. Derived from physical first-principle laws, the dynamics can be described in terms of certain specific physical quantities such as the magnetic flux, voltages, angles, momenta and torques. The complete model can be described by a system of ordinary differential equations (ODE's)
where the flux-current relations are represented by algebraic constraints. The generator rotor circuit is formed by a field circuit and three amortisseur circuits, which is divided in one $d$-axis circuit and two $q$-axis circuits.  The stator is formed by 3-phase windings which are spatially distributed in order to generate 3-phase voltages at machine terminals. For convenience magnetic saturation effects are neglected in the model of the synchronous machine. After applying the $dq0$-transformation $T_{dq0}(\gamma)$ on the ABC-variables with respect to the rotor angle $\gamma$, its dynamics in the $dq0$-reference frame is governed by the following 9th-order system of differential equations \cite{EJCFiaz,kundur,powsysdynwiley}\footnote{See in particular \cite[Chapter 11]{powsysdynwiley} for a detailed derivation of the model \eqref{eq:Mhfullordermodel}.}: 
\begin{subequations}\label{eq:Mhfullordermodel}
\begin{align}
\dot \Psi_d  &= -RI_d-\Psi_q\w -V_d\label{eq:dpsid}\\          
\dot \Psi_q  &= -RI_q+\Psi_d\w -V_q\label{eq:dpsiq}\\          
\dot \Psi_0  &= -RI_0-V_0\label{eq:dpsi0}\\                    
\dot \Psi_f  &= -R_fI_f+V_f\label{eq:dpsif}\\                  
\dot \Psi_g  &= -R_gI_g\label{eq:dphig}\\                      
\dot \Psi_{D}&= -R_{D}I_{D}\\                  
\dot \Psi_{Q}&= -R_{Q}I_{Q}\\                  
\dot \gamma &= \w\label{eq:ddelta}                       \\
 J\dot \w   &= \Psi_qI_d-\Psi_dI_q-d\w+\tau.\label{eq:dwfull}
\end{align}
\end{subequations}
Here $V_{d},V_q,V_0$ are instantaneous external voltages, $\tau$ is the external mechanical torque and $V_f$ is the excitation voltage. The rotor angle $\gamma$, governed by \eqref{eq:ddelta}, is taken with respect to the static \text{ABC}-reference frame, see also Figure \ref{fig:SG}. The quantities $\Psi_{d},\Psi_q,\Psi_0$ are stator winding flux linkages and  $\Psi_{f},\Psi_g,\Psi_D,\Psi_Q$ are the rotor flux linkages respectively and are related to the currents as \cite{powsysdynwiley}
\begin{align}
  \begin{bmatrix}
    \Psi_d\\
    \Psi_f\\
    \Psi_{D}
  \end{bmatrix}&=
\overbrace{  \begin{bmatrix}
    L_d&\kappa M_f&\kappa M_D\\
    \kappa M_f&L_{f}&L_{fD}\\
    \kappa M_D&L_{fD}&L_{D}
  \end{bmatrix}
}^{\mathcal L_d}\begin{bmatrix}
   I_d  \\
    I_f \\
     I_{D}
\end{bmatrix}\label{eq:Md-ind}\\
    \begin{bmatrix}
    \Psi_q\\
    \Psi_g\\
    \Psi_{Q}
  \end{bmatrix}&=
\underbrace{  \begin{bmatrix}
    L_q&\kappa M_g&\kappa M_{Q}\\
    \kappa M_g&L_g&L_{gQ}\\
    \kappa M_{Q}&L_{gQ}&L_{Q}
  \end{bmatrix}}_{\mathclap{\mathcal L_q}}
\begin{bmatrix}
   I_q  \\
   I_g\\
    I_{Q}
\end{bmatrix}\label{eq:Mhq-ind}\\[-1ex]
\Psi_0&=L_0I_0, \label{eq:M0-ind}
\end{align}
where $\kappa=\sqrt{\frac32}$, see also the nomenclature in Section \ref{sec:notation}. 
Note that in the $dq0$-coordinates, the inductor equations can be split up in each of the three axes, resulting into the three completely independent equations \eqref{eq:Md-ind}-\eqref{eq:M0-ind}.  
For a physically relevant model, the inductance matrices $\mathcal L_d,\mathcal L_q\in\mathbb R^{3\times 3}$ are assumed to be positive definite. An immediate observation from \eqref{eq:dpsi0} and \eqref{eq:M0-ind} is that the dynamics associated to the 0-axis is fully decoupled from the rest of the system. Therefore, without loss of generality, we omit this differential equation in the sequel and focus solely on the dynamics in the $d$- and $q$-axes.
\begin{remark}[Additional damper winding]
Many generators, and in particular turbogenerators, have a solid-steel rotor body which acts as a screen in the $q$-axis \cite{powsysdynwiley}. It is convenient to represent this by the additional winding in the $q$-axis represented by the symbol $g$, see \eqref{eq:dphig}. However, for salient-pole synchronous generators, this winding is absent. For completeness, both cases are considered in this paper.
\end{remark}

\subsection{Port-Hamiltonian representation}
\label{sec:port-hamilt-repr-1}
Inspired by the work \cite{EJCFiaz}, it can be shown that full-order model \eqref{eq:Mhfullordermodel} admits a \pH\ representation, see \cite{phsurvey} for a survey. More specifically, by defining the state vector $x=(\Psi_d,\Psi_q,\Psi_f,\Psi_g,\Psi_D,\Psi_Q,\gamma,p), p=J\w$, the $dq$-dynamics of a single synchronous machine can be written in \pH\ form as 
\begin{equation}\label{eq:pHformfull}
\begin{aligned}
\begin{bmatrix}
\dot \Psi_d\\
\dot \Psi_q\\
\dot \Psi_f\\
\dot \Psi_g\\
\dot \Psi_{D}\\
\dot \Psi_{Q}\\
\dot \gamma\\
\dot{ p}
\end{bmatrix}&=
\begin{bmatrix}
  -R     & 0        & 0    & 0    & 0    & 0     & 0 & -\Psi_q \\
  0      & -R       & 0    & 0    & 0    & 0     & 0 & \Psi_d  \\
  0      & 0        & -R_f & 0    & 0    & 0     & 0 & 0       \\
  0      & 0        & 0    & -R_g & 0    & 0     & 0 & 0       \\
  0      & 0        & 0    & 0    & -R_D & 0     & 0 & 0       \\
  0      & 0        & 0    & 0    & 0    & -R_Q  & 0 & 0       \\
  0      & 0        & 0    & 0    & 0    & 0     & 0 & 1       \\
  \Psi_q & -\Psi_d  & 0    & 0    & 0    & 0     & -1& -d      
\end{bmatrix}\nabla H(x)+Gu\\
y&=G^T\nabla H(x)=
\begin{bmatrix}
  I_d\\I_q\\I_f\\\w
\end{bmatrix},\quad  G^T=
\begin{bmatrix}
  1&0&0&0&0&0&0&0\\
  0&1&0&0&0&0&0&0\\
  0&0&1&0&0&0&0&0\\
  0&0&0&0&0&0&0&1
\end{bmatrix}, \quad u=
\begin{bmatrix}
  V_d\\V_q\\V_f\\\tau
\end{bmatrix}.
\end{aligned}
\end{equation}
where the Hamiltonian is given by the sum of the electrical and mechanical energy:
\begin{align*}
H(x)= H_d(x)&+H_q(x)+H_m(x)= \F\begin{bmatrix}
    \Psi_d\\
    \Psi_f\\
    \Psi_{D}
  \end{bmatrix}^T
  \begin{bmatrix}
    L_d&\kappa M_f&\kappa M_D\\
    \kappa M_f&L_{f}&L_{fD}\\
    \kappa M_D&L_{fD}&L_{D}
  \end{bmatrix}^{-1}
  \begin{bmatrix}
    \Psi_d\\
    \Psi_f\\
    \Psi_{D}
  \end{bmatrix}\\&+    \F\begin{bmatrix}
    \Psi_q\\
    \Psi_g\\
    \Psi_{Q}
  \end{bmatrix}^T
  \begin{bmatrix}
    L_q&\kappa M_g&\kappa M_{Q}\\
    \kappa M_g&L_g&L_{gQ}    \\
    \kappa M_{Q}&L_{gQ}&L_{Q}
  \end{bmatrix}^{-1}
\begin{bmatrix}
    \Psi_q\\
    \Psi_g\\
    \Psi_{Q}
  \end{bmatrix}+ \F J^{-1} p^2.
\end{align*}
Here the power-pairs $(V_d,I_d),(V_q,I_q)$ correspond to the external electrical power supplied by the generator. In addition, the power-pair $(V_f,I_f)$ corresponds to the power supplied by the exciter to the synchronous machine. Finally, the pair $(\tau,\w)$ is associated with the mechanical power injected into the synchronous machine. 
As noted from the \pH\ structure of the system \eqref{eq:pHformfull}, it naturally follows that the system is \emph{passive} with respect to the previously mentioned input/output pairs, i.e.,
\begin{align*}
  \dot H\leq V_dI_d+V_qI_q+V_fI_f+\tau\w.
\end{align*}
A crucial observation is that the interconnection 
 structure of the \pH\ system \eqref{eq:pHformfull} depends on the state $x$. This property significantly increases the complexity of a Lyapunov based stability analysis 
 of equilibria that are different from the origin, see  \cite{arjansteginkPerspectiveAnn2016,EJCFiaz,cal-tab,maschke2000energy} for more details on this challenge. 

\section{Model reduction of the synchronous machine}
\label{sec:model-reduct-synchr}

To simplify the analysis of (networked) synchronous machines, it is preferable to consider reduced-order models with decreasing complexity \cite{kundur,sauerpai1998powersystem,powsysdynwiley}. In this section we, following the exposition of \cite{powsysdynwiley}, discuss briefly how several well-known lower order models are obtained from the first-principle model \eqref{eq:Mhfullordermodel}. In each reduction step the underlying assumptions and validity of the reduced-order model is discussed. 

The main assumptions rely on time-scale separation implying that singular perturbation techniques can be used to obtain reduced-order models \cite{ahmed1982reduced}. In particular, in the initial reduction step, this allows the stator windings of the synchronous machine to be considered in quasi steady state. In \cite{kokotovic1980singular} this quasi steady state assumption is validated by the use of iterative time-scale separation.
In doing so, it is assumed that 
the frequency is around the synchronous frequency\footnote{For example, in Europe the synchronous frequency is \SI{50}{Hz} and in the United States it is \SI{60}{\hertz}.} $\w_s$ and that  $\dot \Psi_d,\dot \Psi_q$ are assumed to be small \cite{powsysdynwiley}. 
\begin{assumption}[Operation around $\w\approx \w_s$]\label{ass:ssemfs}
  The synchronous machine is operating around synchronous frequency ($\w\approx \w_s$) and in addition  $\dot \Psi_d$ and $\dot \Psi_q$ are small compared to $-\w\Psi_q$ and $\w\Psi_d$ which implies
  \begin{align}\label{eq:ssemfs}
    \begin{bmatrix}
      V_d\\
      V_q
    \end{bmatrix}\approx
    -
    \begin{bmatrix}
      R&0\\
      0&R
    \end{bmatrix}
         \begin{bmatrix}
           I_d\\
           I_q
         \end{bmatrix}
    +\w_s
    \begin{bmatrix}
      -\Psi_q\\
      \Psi_d
    \end{bmatrix}.
  \end{align}
\end{assumption}
\begin{remark}[Singular perturbation process]
  It is known that during transients $\Psi_d,\Psi_q$ oscillate with high frequency equal to $\w\approx \w_s$ such that $\dot \Psi_d,\dot \Psi_q$ become very large. The validation of the contradicting Assumption \ref{ass:ssemfs} is part of a singular perturbation process where the slow variables are approximated by taking the 
  averaging effect of the fast oscillatory variables \cite{ahmed1982reduced,kokotovic1980singular}. 
\end{remark}
By Assumption \ref{ass:ssemfs}, the two differential equations \eqref{eq:dpsid}, \eqref{eq:dpsiq} corresponding to $\Psi_d,\Psi_q$ are replaced by algebraic equations \eqref{eq:ssemfs}, so that a system of  differential-algebraic equations (DAE's) is obtained \cite{powsysdynwiley}. For many power system studies it is desirable to rephrase and simplify the model \eqref{eq:dpsif}-\eqref{eq:ddelta} together with the algebraic equations \eqref{eq:ssemfs} so that they are in a more acceptable form and easier to interface to the power system network equations. In the following sections, under some additional assumptions based on time-scale separation, we eliminate the two algebraic constraints obtained by putting an equality in \eqref{eq:ssemfs}. 
 Before examining how 
this is done, it is necessary to relate the circuit equations to the flux conditions inside the synchronous machine when it is in the steady state, transient state or the subtransient state.
   
\subsection{Distinction of operation states}
Following the established literature on power systems \cite{kundur,powsysdynwiley,sauerpai1998powersystem,anderson1977}, a distinction between 3 different operation states of the synchronous machine is made. Each  of the 3 characteristic operation states correspond to different stages of rotor screening and a different time-scale \cite{ahmed1982reduced,kokotovic1980singular}, see Figure \ref{fig:operationstates}. 
\begin{figure}
  \centering \includegraphics[width=\linewidth]{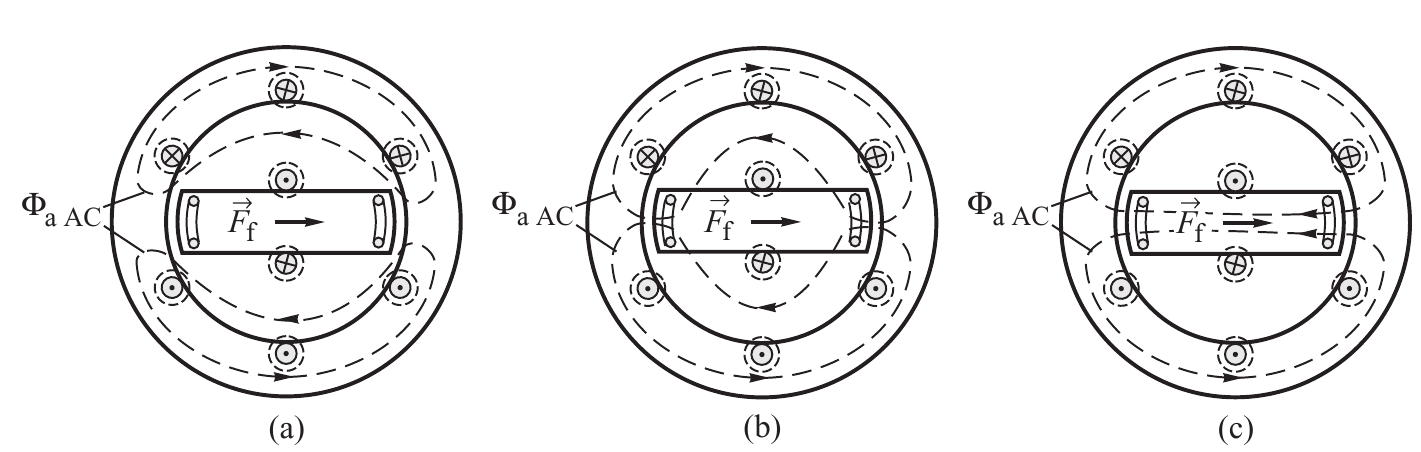}
  \caption[3 different operation states of the synchronous machine.]{ The path of the armature flux in: (a) the subtransient state (screening effect of the damper windings and the field winding); (b) the transient state (screening effect of the field and $g$-damper winding only); (c) the steady state \cite{powsysdynwiley}.
}
\label{fig:operationstates}
\end{figure}

 Immediately after a fault, the current induced in both the rotor field and damper windings forces the armature reaction flux completely out of the rotor to keep the rotor flux linkages constant (this is also referred to as the \emph{Lenz effect}), see Figure \ref{fig:operationstates}a, and the generator is said to be in the \emph{subtransient state} \cite{powsysdynwiley,kundur}. 

As energy is dissipated in the resistance of the rotor windings, the currents maintaining constant rotor flux linkages decay with time allowing flux to enter the windings. As for typical generators the rotor $DQ$-damper winding resistances are the largest, the $DQ$-damper currents are the first to decay, allowing the armature flux to enter the rotor pole face. However, it is still forced out of the field winding and the $g$-damper winding itself, see Figure \ref{fig:operationstates}b. Then the generator is said to be in the \emph{transient state}. 

The field and $g$-winding currents then decay with time to their steady state values allowing the armature reaction flux eventually to enter the whole rotor and assume the minimum reluctance path. Then the generator is in \emph{steady state} as illustrated in Figure \ref{fig:operationstates}c \cite{powsysdynwiley}.

\begin{remark}[Properties of the $g$-damper winding]
  Since the field winding and $g$-damper winding resistances are comparable and are typically much smaller compared to the $DQ$-damper winding resistances, the field winding $f$ and the $g$-damper winding have similar properties in the different operation states.
\end{remark}

\subsubsection{Synchronous machine parameters}
\label{sec:synchr-gener-param}

Depending on which state the synchronous machine is operating in, the effective impedance of the armature coil to any current change will depend on the parameters of the different circuits, their mutual coupling and whether or not the circuits are closed or not \cite{powsysdynwiley}. The inductances and timescales associated with transient and subtransient operation are defined by \cite{powsysdynwiley}
\begin{equation}
\begin{aligned}
L_d'&=L_d-\frac{\kappa ^2M_f^2}{L_f}, \qquad  T_{do}'= \frac{L_f}{R_f},\\
L_q'&=L_q-\frac{\kappa ^2M_g^2}{L_g}, \qquad  T_{qo}'= \frac{L_g}{R_g},\\
L_{d}''&=L_d-\kappa ^2\left[\frac{M_f^2L_D+M_D^2L_f-2M_fM_DL_{fD}}{L_fL_D-L_{fD}^2}\right],\\
L_{q}''&=L_q-\kappa ^2\left[\frac{M_g^2L_Q+ M_Q^2L_g   -2   M_g M_QL_{gQ}}{L_g L_Q-L_{gQ}^2}\right],\\
  T_{do}''&= \frac{1}{R_D}\left(L_D-\frac{L_{fD}^2}{L_f}\right), \quad T_{qo}''= \frac{1}{R_Q}\left(L_Q-\frac{L_{gQ}^2}{L_g}\right).
\end{aligned}\label{eq:subtransientindtime}
\end{equation}
Based on the two-reaction theory of \cite{park1929two}, the corresponding $d$- and $q$-axis reactances for steady state operation ($X_d=\w_sL_d,X_q=\w_sL_q$), transient operation ($X_d'=\w_sL_d',X_q'=\w_sL_q'$) and subtransient operation ($X_d''=\w_sL_d'',X_q''=\w_sL_q''$) are defined.
\begin{remark}[Relation between (sub)transient reactances]
For realistic synchronous machines it holds that $X_d>X_d'>X_d''>0$ and $X_q\geq X_q'>X_q''>0$, where $X_q=X_q'$ holds for a salient-pole synchronous machine (where the $g$-damper winding is absent), see also  \cite[Table 4.3]{powsysdynwiley} and  \cite[Table 4.2]{kundur} for typical values of these reactances. 
\end{remark}

\begin{definition}[Saliency]
The \emph{(sub)transient saliency} is defined as the difference between the (sub)transient reactances, i.e. $X_d'-X_q';(X_d''-X_q'')$. We say that the (sub)transient saliency is negligible if $X_d'=X_q' ; (X_d''=X_q'')$.
\end{definition}
\noindent For both transient and subtransient  state of the machine, different assumptions can be made to obtain the corresponding (differential) equations of the synchronous machine. 

\subsection{Synchronous machine equations}
\label{sec:synchr-gener-equat}

\subsubsection{Transient operation}
\label{sec:modredtranop}
In transient operation state the armature flux has penetrated the damper circuits and the field and $g$ windings screen the rotor body from the armature flux. The damper windings are no more effective ($\dot\Psi_D=\dot\Psi_Q=0$) and thus the damper currents are zero. 
\begin{assumption}[Transient operation]
  During transient operation $I_D=I_Q=0$.
\end{assumption}
From \eqref{eq:Md-ind},  $\Psi_d$ can be expressed in terms of $I_d,\Psi_f$ from which it follows that the internal (transient) and external emfs are related by
\begin{equation}
\begin{aligned}
  V_q&=-RI_q+\w_s \left[I_d\left(L_d-\frac{\kappa ^2M_f^2}{L_{f}}\right)+\frac{\kappa M_f}{L_{f}}\Psi_f\right]\\
&=-RI_q+\w_s L_d' I_d+E_q'=-RI_q+X_d' I_d+E_q'
\end{aligned}\label{eq:VqEq'}
\end{equation}
where the internal emf $E_q'$ is defined by $E_q':=\w_s\left(\frac{\kappa  M_f}{L_f}\right)\Psi_f$. Similarly, from  \eqref{eq:Mhq-ind} we can express $\Psi_q$ in terms of $I_q,\Psi_g$ to obtain 
\begin{align}\label{eq:VdEd'}
  V_d=-RI_d-X_q'I_q+E_d'
\end{align}
where $E_d':=-\w_s\left(\frac{\kappa  M_g}{L_g}\right)\Psi_g$. However, the flux linkages $\Psi_f,\Psi_g$ do not remain constant during transient operation but change slowly as the armature flux penetrates through the windings \cite{powsysdynwiley}.  By substituting \eqref{eq:dpsif}, the differential equation for $E_q'$ is derived as
\begin{equation}
\begin{aligned}
  \dot E_q'&=\w_s \frac{\kappa M_f}{L_f}\dot \Psi_f=\w_s \frac{\kappa M_f}{L_f}(V_f-R_fI_f)
  =\w_s \frac{\kappa M_f}{L_f}(V_f+R_f\frac{\kappa M_fI_d}{L_f})-\frac{R_f}{L_f}E_q'\\
  &=\frac{E_f+(X_d-X_d')I_d-E_q'}{T_{do}'},
\end{aligned}\label{eq:transdyn}
\end{equation}
where we used that $I_D=0$, $T_{do}'=L_f/R_f,$ and the definition $E_f:=\w_s \kappa  M_fV_f/R_f $ for the scaled excitation voltage. In a similar fashion the differential equation of $E_q'$ is derived to obtain
\begin{equation}
\begin{aligned}
  \dot E_d' &=\frac{-(X_q-X_q')I_q-E_d'}{T_{qo}'}.
\end{aligned}\label{eq:transdynEd}
\end{equation}

\subsubsection{Subtransient operation} 
\label{sec:modredsubtranop}
During the subtransient period the rotor damper coils screens both the field winding and the rotor body from changes in the armature flux. The field and $g$ flux linkages $\Psi_f,\Psi_g$ remain constant during this period while the damper winding flux linkages decay with time as the generator moves towards the transient state \cite{powsysdynwiley}. Therefore, we make here a different assumption compared to Section \ref{sec:modredtranop}. 
\begin{assumption}[Subtransient operation]\label{ass:subtranfg}
  During subtransient operation the flux linkages $\Psi_f,\Psi_g$ are constant.
\end{assumption}
Using equation \eqref{eq:Md-ind} one can express $\Psi_d$ in terms of $i_d,\Psi_f,\Psi_D$ to obtain \cite{powsysdynwiley} 
\begin{align*}
  \Psi_d&=L_d''I_d+k_1\Psi_f+k_2\Psi_D,\\
  k_1&=\kappa \cdot\frac{M_fL_D-M_DL_{fD}}{L_fL_D-L_{fD}^2}, \quad k_2=\kappa \cdot\frac{M_DL_f-M_fL_{fD}}{L_fL_D-L_{fD}^2}.
\end{align*}
Together with Assumption \ref{ass:ssemfs} this implies
\begin{align}
  V_q&=-RI_q+\w_s \Psi_d=-RI_q+\w_s L_d''I_d+\w_s k_1\Psi_f+\w_s k_2 \Psi_{D}\nonumber\\
  &=-RI_q+X_d''I_d+E_q''\label{eq:VqEq''}
\end{align}
where $E_q'':=\w_s(k_1\Psi_f+k_2 \Psi_{D})$. 
Similarly for the $q$-axis we obtain 
\begin{align*}
  \Psi_q&=L_q''I_q+k_3\Psi_g+k_4\Psi_Q,\\
  k_3&=\kappa \cdot\frac{M_gL_Q-M_QL_{gQ}}{L_gL_Q-L_{gQ}^2}, \quad k_4=\kappa \cdot\frac{M_QL_g-M_gL_{gQ}}{L_gL_Q-L_{gQ}^2}
\end{align*}
and 
\begin{align}
  V_d&=-RI_q-\w_s \Psi_q=-RI_d-\w_s L_q''I_q-\w_s k_3\Psi_g-\w_s k_4 \Psi_{Q}\nonumber\\
  &=-RI_d-X_q''I_q+E_d''\label{eq:VdEd''}
\end{align}
where $E_d'':=-\w_s(k_3\Psi_g+k_4 \Psi_{Q})$. By eliminating the $I_f,\Psi_d$ from \eqref{eq:Md-ind} and $I_g,\Psi_q$ from \eqref{eq:Mhq-ind} we obtain respectively 
\begin{align*}
  I_D&=\frac{ \kappa  L_{fD} M_fI_d- \kappa  L_f M_DI_d-L_{fD} \Psi_f+L_f \Psi_D}{L_D L_f-L_{fD}^2}\\
  I_Q&=\frac{ \kappa  L_{gQ} M_gI_q- \kappa  L_g M_QI_q-L_{gQ} \Psi_g+L_g \Psi_Q}{L_Q L_g-L_{gQ}^2}.
\end{align*}
Using Assumption \ref{ass:subtranfg}  we find that 
\begin{align}\label{eq:Edq''}
  \dot E_q''=\w_s k_2\dot \Psi_{D}=-\w_s k_2R_DI_D, \qquad \dot E_d''=-\w_s k_4\dot \Psi_Q=\w_s k_4R_QI_Q,
\end{align}
which can be rewritten as 
\begin{align}
  T_{do}''\dot E_q''&=E_q'-E_q''+(X_d'-X_d'')I_d,\label{eq:subtransEq}\\
  T_{qo}''\dot E_d''&=E_d'-E_d''-(X_q'-X_q'')I_q\label{eq:subtransEd}.
\end{align}

\subsubsection{Frequency dynamics}
\label{sec:swingeq}
Recall that the frequency dynamics of the full-order model is described by \eqref{eq:dwfull}:
\begin{align*}
  J\dot \w&=\Psi_qI_d-\Psi_dI_q-d\w+\tau.
\end{align*}
By Assumption \ref{ass:ssemfs}, the latter differential equation is rewritten as
\begin{align*}
  J\dot \w 
           & =-\frac{1}{\w_s}\left(V_dI_d+V_qI_q+R(I_d^2+I_q^2)\right)-d\w+\tau.
\end{align*}
Since the mechanical damping force $F_d=-d\w$ is often very small in large machines, it is neglected in many synchronous machine models \cite{powsysdynwiley,kundur}.
\begin{assumption}[Negligible mechanical damping]\label{ass:nomechdamp}
  The mechanical damping of the synchronous machine is negligible, i.e., $d=0 $. 
\end{assumption}
It is convenient to express the frequency dynamics in terms of the \emph{frequency deviation} with respect to the synchronous frequency $\w_s$. By Assumption \ref{ass:nomechdamp}, the frequency deviation $\Delta \w:=\w-\w_s$ is governed by the differential equation
\begin{equation}
\begin{aligned}\label{eq:freqdynt}
 J \Delta\dot { \w}&=-\frac{1}{\w_s}\left(V_dI_d+V_qI_q+R(I_d^2+I_q^2)\right)+\tau.
\end{aligned}
\end{equation}
After multiplying \eqref{eq:freqdynt} by the synchronous frequency $\w_s$ one obtains
\begin{equation}
\begin{aligned}\label{eq:freqdynpower}
  M\Delta\dot {\w}&=-\left(V_dI_d+V_qI_q+R(I_d^2+I_q^2)\right)+\w_s\tau=-P_e+P_m,
\end{aligned}
\end{equation}
where it is common practice to define the quantity $M:=\w_s J$ \cite{kundur,powsysdynwiley}. 
Here the mechanical power injection is denoted by $P_m=\w_s\tau$ and the electrical power $P_e$ produced by the synchronous generator is equal to 
\begin{align*}
  P_e=V_dI_d+V_qI_q+R(I_d^2+I_q^2).
\end{align*}

\begin{remark}[Alternative formulation of frequency dynamics]
  Note that  by equations \eqref{eq:VqEq''} and \eqref{eq:VdEd''} the electrical power $P_e$ produced by the synchronous generator alternatively takes the form 
\begin{align}\label{eq:Pe}
  P_e=E_d''I_d+E_q''I_q+(X_d''-X_q'')I_dI_q
\end{align}
such that the differential equation \eqref{eq:freqdynpower} can be rewritten as
  \begin{equation}
  \begin{aligned}\label{eq:subtransfreq}
    M\Delta\dot { \w}&=-E_d''I_d-E_q''I_q-(X_d''-X_q'')I_dI_q+P_m.
  \end{aligned}
\end{equation}
\end{remark}

\subsection{Synchronous machine models}
\label{sec:synchr-mach-models}


Based on the results established in Section \ref{sec:synchr-gener-equat}, several generator models with decreasing  complexity and accuracy are developed. In each model reduction step, the validity and assumptions made in the corresponding model are discussed.

\subsubsection{Sixth-order model}
\label{sec:sixth-order-model}
By combining the equations derived in Section \ref{sec:synchr-gener-equat},  
a sixth-order model describing the synchronous generator is obtained. In particular, by \eqref{eq:transdyn}, \eqref{eq:transdynEd}, \eqref{eq:subtransEq}, \eqref{eq:subtransEd} and \eqref{eq:subtransfreq} we obtain the following system of ordinary differential equations describing the generator dynamics \cite{powsysdynwiley}:
\begin{subequations}\label{eq:SG6order}
\begin{align}
\dot \delta&=\Delta \w \\
  M\Delta \dot  \w&=P_{m}-E_d''I_d-E_q''I_q-(X_d''-X_q'')I_dI_q\label{eq:wdyn6o}\\
T_{do}'\dot E_q'&=E_{f}-E_q'+I_d(X_d-X_d') \label{eq:Eqq}\\
T_{qo}'\dot E_d'&=-E_d'-I_q(X_q-X_q')\label{eq:Edd}\\
T_{do}''\dot E_q''&=E_q'-E_q''+I_d(X_d'-X_d'')\label{eq:Eqqq}\\
T_{qo}''\dot E_d''&=E_d'-E_d''-I_q(X_q'-X_q''),\label{eq:Eddd}
\end{align}
\end{subequations}
where $\delta(t):=\gamma(t)-\w_s t$ represents the rotor angle with respect to the synchronous rotating reference frame. By equations \eqref{eq:VqEq''} and \eqref{eq:VdEd''}  the internal and external voltages of the synchronous generator are related by
\begin{align}\label{eq:VE}
  \begin{bmatrix}
    V_d\\
    V_q
  \end{bmatrix}=
  \begin{bmatrix}
    E_d''\\E_q''
  \end{bmatrix}-
  \begin{bmatrix}
    R&X_q''\\-X_d''&R
  \end{bmatrix}
 \begin{bmatrix}
   I_d\\I_q
 \end{bmatrix}.
\end{align}
It is worth noting the similar structure of these (differential) equations.  The equation \eqref{eq:VE} and the right hand side of \eqref{eq:Eqq}-\eqref{eq:Eddd}  relates to the equivalent $d$- or $q$-axis generator circuits, with the resistances neglected, as shown in Figure \ref{fig:gencirc}. In particular, the algebraic equation \eqref{eq:VE} corresponds to the right-hand side of Figure \ref{fig:gencirc}. In addition, the subtransient dynamics \eqref{eq:Eqqq}, \eqref{eq:Eddd} corresponds to the center reactances $X_d'-X_d'', X_q'-X_q''$ illustrated in Figure \ref{fig:gencirc} and the transient dynamics \eqref{eq:Eqq}, \eqref{eq:Edd} corresponds to the left-hand side of Figure \ref{fig:gencirc}. Observe that there is no additional voltage in the $q$-axis due to the absence of a field winding on this axis. 

\begin{figure}
  \centering
  \begin{circuitikz}[scale=1]
    \draw[] (0,0) to [voltage source,*-*,v^=$\overline E_{f}$] ++(0,2)
    to [american inductor,l^=$ j(X_d-X_d') $,*-*] ++(2.2,0) node(14){}
    to [american inductor,l^=$ j(X_d'-X_d'') $,*-*] ++(2.2,0)
    node(15){} to [american inductor,l^=$ jX_d''
    $,*-o,i>=$\overline
    I_{d}$]
    ++(2,0) node(1){}; \draw[very thick] (0,0) -- (6.4,0) node(5){};
    \draw[] (2.2,0) node(11){} (4.4,0) node(12){} ;
    \draw[-latex,thick] (11) --node[right]{$\overline E_q'$} (14);
    \draw[-latex,thick] (12) --node[right]{$\overline E_q''$} (15);
    \draw[->,rotate=0,thick] (3.4,0.6) arc (-90:180:0.4) node[right]{$T_{do}''$};
    \draw[->,rotate=0,thick] (1.1,0.6) arc (-90:180:0.4) node[right]{$T_{do}'$};   
    \foreach \i in {3.3} { \draw[] (0,0-\i) to[*-*] ++(0,2) to
      [american inductor,l^=$ j(X_q-X_q') $,*-*] ++(2.2,0) node(14){}
      to [american inductor,l^=$ j(X_q'-X_q'') $,*-*] ++(2.2,0)
      node(15){} to [american inductor,l^=$ jX_q''
      $,*-o,i>=$\overline
      I_{q}$]
      ++(2,0) node(2){}; \draw[very thick] (0,-\i) to (6.4,-\i)
      node(6){};
      \draw[] (2.2,-\i) node(11){} (4.4,-\i) node(12){} ;
      \draw[-latex,thick] (11) --node[right]{$\overline E_d'$} (14);
      \draw[-latex,thick] (12) --node[right]{$\overline E_d''$} (15);
      \draw[->,rotate=0,thick] (3.4,0.6-\i) arc (-90:180:0.4) node[right]{$T_{qo}''$};
      \draw[->,rotate=0,thick] (1.1,0.6-\i) arc (-90:180:0.4) node[right]{$T_{qo}'$};;
    } \draw[-latex,thick] (5) to node[right]{$\overline V_q$} (1);
    \draw[-latex,thick] (6) to node[right]{$\overline V_d$}(2);
  \end{circuitikz}
  \caption{The generator equivalent circuits for both $dq$-axes in case the stator winding resistance $R$ is neglected \cite{powsysdynwiley}. 
  }
\label{fig:gencirc}
\end{figure}
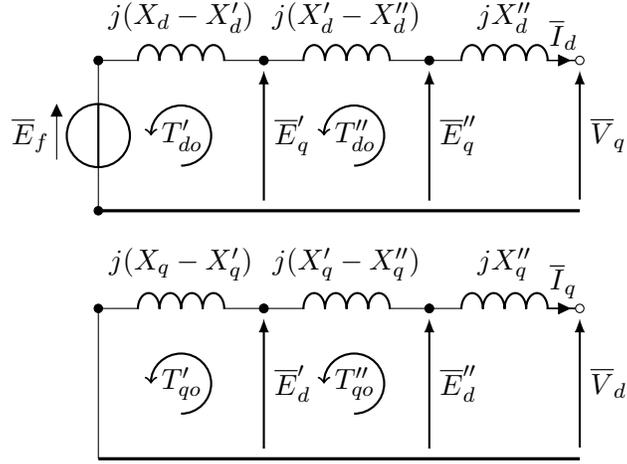

\subsubsection{Fifth-order model}
In a salient-pole generator the laminated rotor construction prevent eddy currents flowing in the rotor body such that there is no screening in the $q$-axis implying that $X_q=X_q'$ \cite{powsysdynwiley}. In that case the $g$-winding is absent in the full-order model \eqref{eq:Mhfullordermodel}. Consequently, $E_d'$ is absent  so that the fifth-order model becomes
\begin{equation}
\begin{aligned}\label{eq:SG5order}
\dot \delta&=\Delta \w \\
  M\Delta \dot  \w&=P_{m}-E_d''I_d-E_q''I_q-(X_d''-X_q'')I_dI_q\\
T_{do}'\dot E_q'&=E_{f}-E_q'+I_d(X_d-X_d')\\
T_{do}''\dot E_q''&=E_q'-E_q''+I_d(X_d'-X_d'')\\
T_{qo}''\dot E_d''&=-E_d''-I_q(X_q'-X_q'').
\end{aligned}
\end{equation}

\subsubsection{Fourth-order model}
In this model the subtransient dynamics of the sixth-order model induced by the damper windings is neglected. This is motivated by the fact that $T_{do}''\ll T_{do}', T_{qo}''\ll T_{qo}'$. Therefore the dynamics corresponding with $E_q'',E_d''$ is at much faster time scale compared to the  $E_q',E_d'$ dynamics. As a result, at the slower time-scale we obtain the quasi steady state condition \cite{ahmed1982reduced}:
\begin{equation}
\begin{aligned}
  E_q''&=E_q'+I_d(X_d'-X_d'')\\
  E_d''&=E_d'-I_q(X_q'-X_q'').
\end{aligned}\label{eq:subtrans0}
\end{equation}
Substitution of the latter algebraic equations in the remaining four differential equations yields the fourth-order model
\begin{equation}
\begin{aligned}\label{eq:SG4order}
\dot \delta&=\Delta \w \\
  M\Delta \dot  \w&=P_{m}-D\Delta\w-E_d'I_d-E_q'I_q-(X_d'-X_q')I_dI_q\\
T_{do}'\dot E_q'&=E_{f}-E_q'+I_d(X_d-X_d')\\
T_{qo}'\dot E_d'&=-E_d'-I_q(X_q-X_q').
\end{aligned}
\end{equation}
\begin{remark}[Transient operation]
  Note that \eqref{eq:subtrans0} together with \eqref{eq:VE} also implies \eqref{eq:VqEq'} and \eqref{eq:VdEd'} as expected since the subtransient dynamics is neglected. 
\end{remark}
As the damper windings are ignored, the air-gap power appearing in the frequency dynamics neglects the asynchronous torque produced by the damper windings. To compensate the effects of the damper windings a linear asynchronous damping power $D \Delta \w $ with damping constant $D>0$ is introduced \cite{powsysdynwiley}. However, more accurate nonlinear approximations of the damping power exist as well, see \cite[Chapter 5.2]{powsysdynwiley}. 

\subsubsection{Third-order model}
\label{sec:third-order-model}
Starting from the fourth-order model, we make here the same assumptions as done in the transition from the sixth-order model to the fifth-order model ($E_d'=0$) so that the third-order model, which also referred to as the \emph{flux-decay model} or \emph{one-axis model} \cite{kundur}, is given by 
\begin{subequations}\label{eq:SG3order}
\begin{align}
\dot \delta&=\Delta \w \\
  M\Delta \dot  \w&=-D\Delta \w+P_{m}-E_q'I_q-(X_d'-X_q')I_dI_q\label{eq:freqdyn3o}\\
T_{do}'\dot E_q'&=E_{f}-E_q'+I_d(X_d-X_d').
\end{align}
\end{subequations}

\subsubsection{Second-order classical model}
\label{sec:second-order-model}
The second-order model is derived from the fourth-order (or third-order) model by assuming that the internal emfs $E_q',E_d'$ are constant \cite{kundur,anderson1977,powsysdynwiley}. This can be validated if the timescales $T_{qo}',T_{do}'$ are large (of the order of a few seconds) so that the internal emfs $E_q',E_d'$ can be approximated by a constant (on a bounded time interval) provided that $E_f,I_d,I_q$ do not change much. From this assumption, a constant voltage behind transient reactance model is obtained which is commonly referred to as the \emph{constant flux linkage model} or \emph{classical model} \cite{kundur,powsysdynwiley,anderson1977}: 
\begin{equation}
\begin{aligned}\label{eq:SG2order}
\dot \delta&=\Delta \w \\
  M\Delta \dot  \w&=-D\Delta \w+P_{m}-E_q'I_q-E_d'I_d-(X_d'-X_q')I_dI_q
\end{aligned}
\end{equation}
The assumption that the changes in $dq$-currents and the internal emfs are small implies that only generators located a long way from the point of the disturbance should be represented by the classical model \cite{powsysdynwiley}. In addition, since the assumption that $E_q',E_q'$ is constant is only valid on a limited time-interval,  
the classical model is only valid for analyzing the \emph{first swing stability} \cite{anderson1977}. Indeed, in for example \cite{caliskan2015uses} it was shown that the second-order swing equations \eqref{eq:SG2order} are not valid for asymptotic stability analyses. 



\section{Multi-machine models}
\label{sec:multi-machine-models}
To obtain a representation of the power grid, we consider a multi-machine network.  For simplicity we consider the case that each node in the network represents a synchronous machine, that is, each node represents either a synchronous generator, or a synchronous motor. In addition, we assume that the stator winding resistances and the resistances in the network are negligible. This assumption is valid for networks with high voltage transmission lines where the line resistances are negligible.
\begin{assumption}[Inductive lines]\label{ass:purelyinductivelines}
  The network is considered to be purely inductive and the stator winding resistances are negligible, i.e., $R=0$. 
\end{assumption}

In this section the multi-machine models starting from the sixth-, third-, and second-order models for the synchronous generator are established.  The derivations of the fourth- and fifth-order multi-machine models are omitted as these are very similar to ones presented in this section. 
To obtain reduced-order multi-machine models,  the equations for the nodal currents in the network are derived which are then substituted in the single generator models reformulated in Section \ref{sec:synchr-mach-models}.

\subsection{Sixth-order multi-machine model}
For the sixth (and fifth) order model(s) it is convenient to make the following assumption which is valid for synchronous generators with damper windings in both $d$- and $q$-axes \cite{powsysdynwiley}.
\begin{assumption}[$X_{di}''=X_{qi}''$]\label{ass:subtranssal}
  For each synchronous machine in the network, the subtransient saliency is negligible, i.e., $X_{di}''=X_{qi}'' \ \  \forall i\in\V$. 
\end{assumption}
By Assumption \ref{ass:subtranssal}, the second term of the electrical power \eqref{eq:Pe} appearing in the frequency dynamics \eqref{eq:wdyn6o} vanishes. Moreover, the assumption of $X_d'' = X_q''$ allows the two individual $d$- and $q$-axis circuits in Figure \ref{fig:gencirc} to be replaced by one equivalent circuit, see Figure \ref{fig:emfsubtran}. As a result, all the voltages, emfs and currents are phasors in the synchronous rotating reference frame of rather than their components resolved along the $d$- and $q$-axes. An important advantage of this is that the generator reactance may be treated in a similar way as the reactance of a transmission line, as we will show later. This has particular importance for multi-machine systems when combining the algebraic equations describing the generators and the network \cite{powsysdynwiley}.

\begin{figure}[h]
\begin{center}
\begin{circuitikz}
\draw[]  (0,0)   to [european voltage source,*-*,v_=$\overline E_i''$]  ++(0,2)
 to [american inductor,l^=$ jX_{di}'' $,*-o,i=$\overline I_i$] ++(2.5,0) node(1){};
\draw[very thick] (0,0) to[short] (2.5,0) node(2){};
\draw[thick,-latex] (2) to node[left]{$\overline V_i$} (1);
\end{circuitikz}
\end{center}
\caption{Subtransient emf behind a subtransient reactance.} 
\label{fig:emfsubtran}
\end{figure}
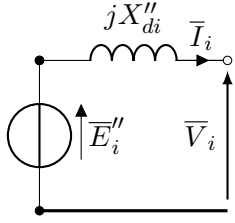
As illustrated in Figure \ref{fig:emfsubtran}, the internal and external voltages are related to each other by 
\begin{align} \label{eq:EVrelation}
\overline E_i''&=\overline V_i+jX_{di}''\overline I_i, \qquad \forall i\in\V.
\end{align}  
Consider a power network where each node $i\in\V=\{1,2,\ldots,n\}$ represents a synchronous machine and each edge $(i,k)\in\E$ a transmission line, see Figure \ref{fig:indlinesubtrans} for a two-node case. 
\begin{figure}[h]
\begin{center}
\begin{circuitikz}
\draw[]  (0,0)   to [european voltage source,*-*,v_=$\overline E_{i}''$]  ++(0,2)
 to [american inductor,l^=$j X_{di}'' $,*-*] ++(2.5,0) node(1){}
 to [american inductor,l^=$ jX_{T_{ik}} $,*-*,i=$\overline I_{ik}$] ++(2.5,0) node(2){}
 to [american inductor,l^=$j X_{dk}'' $,*-*] ++(2.5,0) 
 to [european voltage source,v_<=$\overline E_{k}''$] ++(0,-2)  
to[short,*-*] ++(-7.5,0) 
;
\draw[] (2.5,0) node(3){}; \draw[] (5,0) node(4){};
\draw[thick,-latex] (3) to node[left]{$\overline V_i$} (1); 
\draw[thick, -latex] (4) to node[right]{$\overline V_k$} (2);
\end{circuitikz}
\end{center}
\caption{Interconnection of two synchronous machines governed by the 5th or 6th order model by a purely inductive transmission line with reactance $X_{T_{ik}}$.} 
\label{fig:indlinesubtrans}
\end{figure}
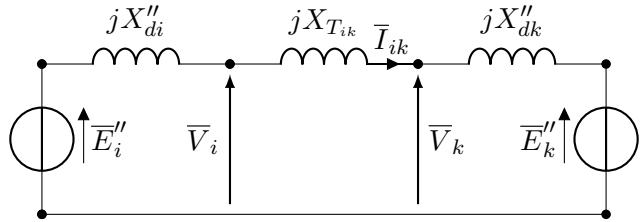
To derive the algebraic equations associated with the network, we assume that the network operates at steady state. Under this assumption, the network equations take the form
\begin{align*}
  \overline I_s=Y\overline V_s=\mathcal Y \overline E''_s
\end{align*}
where $\overline I_s,\overline V_s,\overline E_s''\in\mathbb C^n$ represent the nodal current and  external/internal voltage phasors with respect to the synchronous rotating reference frame and $Y\in\mathbb C^{n\times n}$ is the admittance matrix of the network. The admittance matrix $\mathcal Y\in\mathbb C^{n\times n}$ is obtained by adding the reactances $X_{di}'',i\in\V$ to the transmission line reactances, i.e.,  $\mathcal Y$   takes the form $\mathcal Y_{ii}=G_{ii}+jB_{ii}, \mathcal Y_{ik}=-G_{ik}-jB_{ik},i\neq k$ where the susceptances are given by \cite{schiffer2016survey}
\begin{equation}
\begin{aligned}
  B_{ik}&=
  \begin{cases}
    0&\text{if nodes $i$ and $k$ are not connected} \\
    -\frac1{X_{ik}}&\text{if nodes $i$ and $k$ are  connected}
  \end{cases}\\
B_{ii}&=\sum_{k\in\mathcal N_i}B_{ik}
\end{aligned}\label{eq:susceptances}
\end{equation}
and where $X_{ik}:=X_{T_{ik}}+X_{di}''+X_{dk}''$ is the total reactance between the subtransient voltage sources as illustrated in Figure \ref{fig:indlinesubtrans}. As we assumed purely inductive lines, see Assumption \ref{ass:purelyinductivelines},  the conductance matrix equals the zero matrix and thus $G_{ik}=0 \ \forall i,k\in\V$.
We note that in the derivations in Section \ref{sec:model-reduct-synchr} the currents $\overline I=V_q+jV_d$ and internal voltages $\overline E''=E_q''+jE_d''$ are expressed with respect to the \emph{local} $dq0$-reference frame of the synchronous machine. Thus, according to \eqref{eq:phasor-trans},  $\overline I_s=\diag(e^{-j(\w_st-\gamma_i)})\overline I=\diag(e^{j\delta_i})\overline I$ and similarly $\overline E_s''=\diag(e^{j\delta_i})\overline E''$. Consequently,  
\begin{align}\label{eq:hatIV}
  \overline I&=\diag(e^{-j\delta_i})\mathcal Y \diag(e^{j\delta_i})  \overline E'',
\end{align}
where $\overline I=\col(\overline I_1,\ldots,\overline I_n), \overline E''=\col(\overline E_1'',\ldots,\overline E_n'')$. 
Then the $dq$-current phasor at node $i$ takes the form
\begin{align}\label{eq:Ii}
\overline I_i&=\mathcal Y_{ii}\overline E_i''+\sum_{k\in\mathcal N_i}\mathcal Y_{ik}e^{-j\delta_{ik}}\overline E_k''.
\end{align}
 Using the phasor representation $\overline E_i''=E_{qi}''+jE_{di}'', \overline I_i =I_{qi}+jI_{di}$, and equating both the real and imaginary part of equation \eqref{eq:Ii}, we obtain after rewriting
\begin{equation}
\begin{alignedat}{2}\label{eq:currentsmulti}
  I_{di}&=&&B_{ii}E_{qi}''-\sum_{k\in\mathcal N_i}\left[B_{ik}(E_{dk}''\sin \delta_{ik}+E_{qk}''\cos \delta_{ik})   \right],\\
I_{qi}&=\ &-&B_{ii}E_{di}''-\sum_{k\in\mathcal N_i}\left[B_{ik}(E_{qk}''\sin \delta_{ik}-E_{dk}''\cos \delta_{ik})\right].
\end{alignedat}
\end{equation}
\begin{remark}[Nonzero transfer conductances]\label{rem:Idqresistance}
Compared to \eqref{eq:currentsmulti},  a slightly more complicated expression for the $dq$-currents  can be derived in the more general case where the transfer conductances are nonzero, see e.g. \cite{schiffer2016survey}. 
\end{remark}
By substituting the network equations \eqref{eq:currentsmulti} into the sixth-order model of the synchronous machine derived in Section \ref{sec:sixth-order-model}, the multi-machine model \eqref{eq:multimach6} is obtained. 
A subscript $i$ is added to the model \eqref{eq:SG6order} to indicate that this is the model of synchronous machine $i\in\V$.
\begin{align}
\dot \delta_i&=\Delta \w_i \nonumber\\
  M_i\Delta \dot \w_i&=P_{mi}+\sum_{k\in\mathcal N_i}B_{ik}\Big[(E_{di}''E_{dk}'' +E_{qi}''E_{qk}'')\sin \delta_{ik}+(E_{di}''E_{qk}''-E_{qi}''E_{dk}'')\cos \delta_{ik}\Big]\nonumber\\
T_{doi}'\dot E_{qi}'&=E_{fi}-E_{qi}'+(X_{di}-X_{di}')(B_{ii}E_{qi}''-\sum_{k\in\mathcal N_i}\left[B_{ik}(E_{dk}''\sin \delta_{ik}+E_{qk}''\cos\delta_{ik})   \right])\nonumber\\
T_{qoi}'\dot E_{di}'&=-E_{di}'+(X_{qi}-X_{qi}')(B_{ii}E_{di}''-\sum_{k\in\mathcal N_i}\left[B_{ik}(E_{dk}''\cos\delta_{ik}-E_{qk}''\sin\delta_{ik})\right]) \label{eq:multimach6}\\
T_{doi}''\dot E_{qi}''&=E_{qi}'-E_{qi}''+(X_{di}'-X_{di}'')(B_{ii}E_{qi}''-\sum_{k\in\mathcal N_i}\left[B_{ik}(E_{dk}''\sin\delta_{ik}+E_{qk}''\cos\delta_{ik})   \right])\nonumber\\
T_{qoi}''\dot E_{di}''&=E_{di}'-E_{di}''+(X_{qi}'-X_{qi}'')(B_{ii}E_{di}''-\sum_{k\in\mathcal N_i}\left[B_{ik}(E_{dk}''\cos\delta_{ik}-E_{qk}''\sin\delta_{ik})\right])\nonumber
\end{align}
The electrical power $P_{ei}$ produced by synchronous machine $i$ is obtained from \eqref{eq:Pe} and \eqref{eq:currentsmulti}, and is given by
\begin{equation}
\begin{aligned}
  P_{ei}&=E_{di}''I_{di}+E_{qi}''I_{qi}\\
&=\sum_{k\in\mathcal N_i}\underbrace{-B_{ik}\Big[(E_{di}''E_{dk}'' +E_{qi}''E_{qk}'')\sin \delta_{ik}+(E_{di}''E_{qk}''-E_{qi}''E_{dk}'')\cos \delta_{ik}\Big]}_{P_{ik}}.
\end{aligned}\label{eq:Pei6order}
\end{equation}
\begin{remark}[Energy conservation]\label{rem:Pei0}
 Since the transmission lines are purely inductive by assumption, there are no energy losses in the transmission lines implying that the following energy conservation law holds: $P_{ik}=-P_{ki}$ where $P_{ik}$ given in \eqref{eq:Pei6order} represents the power transmission from node $i$ to node $k$. In particular, we also have $\sum_{i\in\V}P_{ei}=0$ with $P_{ei}$ is given by \eqref{eq:Pei6order}.
\end{remark}
\begin{remark}[Including resistances]\label{rem:nonzeroG}
  While in the above model the resistances of the network and the stator windings are neglected, the model easily extends to the case of nonzero resistances. This can be done following the same procedure as before but instead substituting the more complicated expression for the currents $I_{di},I_{qi}$, see Remark \ref{rem:Idqresistance}. 
\end{remark}

\subsection{Third-order multi-machine model}
The derivation of the third-order multi-machine models proceeds along the same lines as for the sixth-order model. For similar reasons as for the sixth- and fifth-order models, it is convenient for the 2nd, 3rd and 4th order multi-machine models to assume that the transient saliency is negligible. 
\begin{assumption}[$X_{di}'=X_{qi}'$]\label{ass:transsalneg}
  The transient saliency is negligible: $X_{di}'=X_{qi}' \ \  \forall i\in\V$. 
\end{assumption}
By making the \emph{classical assumption} that $X_d'=X_q'$, the second term of the electrical power appearing in the frequency dynamics \eqref{eq:freqdyn3o} vanishes \cite{powsysdynwiley}. In addition, 
\begin{figure}
\begin{center}
\begin{circuitikz}
\draw[]  (0,0)   to [european voltage source,*-*,v_=$\overline E_i'$]  ++(0,2)
 to [american inductor,l^=$ jX_{di}' $,*-o,i=$\overline I_i$] ++(2.5,0) node(1){};
\draw[very thick] (0,0) to[short] (2.5,0) node(2){};
\draw[thick,-latex] (2) to node[left]{$\overline V_i$} (1);
\end{circuitikz}
\end{center}
\caption{Single generator equivalent circuit in case the transient saliency is neglected \cite{powsysdynwiley}}
\label{fig:oneeqcirc}
\end{figure}
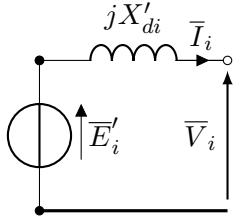
the assumption of $X_d' = X_q'$ allows the separate $d$ and $q$-axis circuits shown in Figure \ref{fig:gencirc} to be replaced by one simple equivalent circuit, see Figure \ref{fig:oneeqcirc}, representing a transient voltage source behind a transient reactance. 
\begin{remark}[Negligible transient saliency]
Although there is always some degree of transient saliency implying that $X_{di}'\ne X_{qi}'$, it should be noted that if the network reactances are relatively large, then the effect of the transient saliency on the power network dynamics is negligible making Assumption \ref{ass:transsalneg} acceptable \cite{powsysdynwiley}. 
\end{remark}


Similar as before, the interconnection of two synchronous machines can be represented as in Figure \ref{fig:indlinetrans}. 
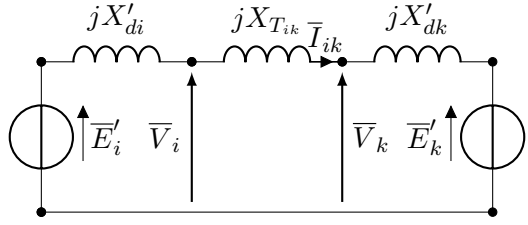
\begin{figure}[h]
\begin{center}
\begin{circuitikz}
\draw[]  (0,0)   to [european voltage source,*-*,v_=$\overline E_{i}'$]  ++(0,2)
 to [american inductor,l^=$j X_{di}' $,*-*] ++(2,0) node(1){}
 to [american inductor,l^=$ jX_{T_{ik}} $,*-*,i=$\overline I_{ik}$] ++(2,0) node(2){}
 to [american inductor,l^=$j X_{dk}' $,*-*] ++(2,0) 
 to [european voltage source,v_<=$\overline E_{k}'$] ++(0,-2)  
to[short,*-*] ++(-6,0) 
;
\draw[] (2,0) node(3){}; \draw[] (4,0) node(4){};
\draw[thick,-latex] (3) to node[left]{$\overline V_i$} (1); 
\draw[thick, -latex] (4) to node[right]{$\overline V_k$} (2);
\end{circuitikz}
\end{center}
\caption{Interconnection of two synchronous machines governed by the 2nd, 3rd or 4th order model by a purely inductive transmission line with reactance $X_{T_{ik}}$.} 
\label{fig:indlinetrans}
\end{figure}
As illustrated in this figure, the internal and external voltages are related to each other by \cite{powsysdynwiley}
\begin{align} \label{eq:EVrelation}
\overline E_i'&=\overline V_i+jX_{di}'\overline I_i, \qquad \forall i\in\V.
\end{align}  
 The algebraic equations associated with the network amount to \cite{schiffer2016survey}
\begin{align}\label{eq:IYEp}
  \overline I&=\diag(e^{-j\delta_i})\mathcal Y \diag(e^{j\delta_i})  \overline E',
\end{align}
resulting in a similar expression for the $dq$-currents as for the sixth-order model:
\begin{equation}
\begin{alignedat}{2}\label{eq:currentsmultitrans}
  I_{di}&=&&B_{ii}E_{qi}'-\sum_{k\in\mathcal N_i}\left[B_{ik}(E_{dk}'\sin \delta_{ik}+E_{qk}'\cos \delta_{ik})   \right],\\
I_{qi}&=\ &-&B_{ii}E_{di}'-\sum_{k\in\mathcal N_i}\left[B_{ik}(E_{qk}'\sin \delta_{ik}-E_{dk}'\cos \delta_{ik})\right].
\end{alignedat}
\end{equation}
By using the third-order model of the synchronous machine \eqref{eq:SG3order},  the network equations \eqref{eq:currentsmultitrans}, and the fact that that $E_{di}'=0$ for the third-order model, the flux-decay (or one-axis) multi-machine model is obtained.
\begin{equation}
\begin{aligned}
\dot \delta_i&=\Delta \w_i \\
 M_i\Delta \dot \w_i&=P_{mi}-D_i\Delta\w_i+\sum_{k\in\mathcal N_i}B_{ik}E_{qi}'E_{qk}'\sin \delta_{ik}\\
T_{doi}'\dot E_{qi}'&=E_{fi}-E_{qi}'+(X_{di}-X_{di}')(B_{ii}E_{qi}'-\sum_{k\in\mathcal N_i}B_{ik} E_{qk}'\cos \delta_{ik}))
\end{aligned}\label{eq:multimach3}
\end{equation}
It is observed that, similar as for the sixth-order multi-machine model \eqref{eq:multimach6}, Remark \ref{rem:Pei0} and Remark \ref{rem:nonzeroG} also hold  for the third-order model \eqref{eq:multimach3}. 

\subsection{The classical multi-machine network}
\label{sec:class-multi-mach}
The derivation of the classical second-order swing equations takes a slightly different approach compared to the multi-machine models obtained previously. For completeness, the derivation of the second-order multi-machine model with $RL$-transmission lines is given in this section. 

Suppose that Assumption \ref{ass:transsalneg} holds. Let the transient voltage phasor be represented as $\overline E'_i=e^{j\alpha_i}|\overline{E}'_i|$, then by \eqref{eq:IYEp} we have 

\begin{align*}
  \overline I_i&
  =\mathcal Y_{ii}e^{j\alpha_i} |\overline{E}'_i|+\sum_{k\in\mathcal N_i}\mathcal Y_{ik}e^{-j\delta_{ik}}e^{j\alpha_k} |\overline{E}'_k|, \qquad \forall i\in\V.
\end{align*}
By defining the angles\footnote{Note that the angle $\theta_i$ represents the \emph{voltage angle} of generator $i$ with respect to the synchronous rotating reference frame.} $\theta_i:=\delta_i+\alpha_{i}$ it can be shown that the electrical power supplied by the synchronous machine amounts to 
\begin{align*}
P_{ei}&=\Re(  \overline E_i'\overline I_i^*)=\Re(  \overline E_i'^*\overline I_i)=\Re\Big(\mathcal Y_{ii}|\overline{E}_i'|^2+\sum_{k\in\mathcal N_i}\mathcal Y_{ik}e^{-j(\delta_{ik}+\alpha_{ik})} |\overline{E}_i'||\overline{E}_k'|\Big)\\
&=G_{ii}|\overline E'_i|^2-\sum_{k\in\mathcal N_i}(G_{ik}\cos \theta_{ik} +B_{ik}\sin \theta_{ik}) |\overline E'_i||\overline E'_k|.
\end{align*}
It is convenient to express the system dynamics in terms of the voltage angles $\theta_i$. By noting that $\alpha_i$ is constant\footnote{Note that for the third-order model $\alpha_i=0$ implying that in this case $\theta_i$ is equal to the rotor angle $\delta_i$ with respect to the synchronous rotating reference frame.} it follows that $\dot \theta_i=\dot \delta_i=\Delta\w_i$. Hence, the multi-machine classical model with nonzero transfer conductances is described by 
\begin{equation}
\begin{aligned}\label{eq:SG2multimach}
\dot \theta_i&=\Delta \w_i\\
  M_i\Delta \dot  \w_i&=-D_i\Delta \w_i+P_{mi}-G_{ii}|\overline E'_i|^2\\
&+\sum_{k\in\mathcal N_i}(G_{ik}\cos \theta_{ik} +B_{ik}\sin \theta_{ik}) |\overline E'_i||\overline E'_k|, \qquad i\in\V.
\end{aligned}
\end{equation}
\begin{remark}[Purely inductive network]
  Note that in a purely inductive network $G=0$ and $B_{ik}\leq 0$ for all $i,k$.   The resulting multi-machine network, commonly referred to as the \emph{swing equations}, is often used in power network stability studies, see e.g. \cite{LHMNLC,AGC_ACC2014,zhangpapaautomatica,swing-claudio}. 
\end{remark}
\begin{remark}[Load nodes]
  In the multi-machine models constructed in this section it is assumed that each node in the network represents a synchronous machine. However, a more realistic model of a power network can be obtained by making a distinction between generator and load nodes \cite{alv_meng_power_coupl_market,firstSPM}. This is beyond the scope of the present paper. Instead, we assume that some synchronous machines act as \emph{synchronous motors} for which the injected mechanical power is \emph{negative}.  
\end{remark}

\section{Energy functions}
\label{sec:energy-analysis}

When analyzing the stability of a synchronous machine (or a multi-machine network) it is desired to search for a suitable Lyapunov function. Often the physical energy stored in the system can be used as a Lyapunov function for the zero-input case. In this section we derive the energy functions of the reduced order models of the synchronous machine. In addition, the energy functions corresponding to the transmission lines are obtained. 

\subsection{Synchronous machine}
The physical energy stored in a synchronous machine consists of both an electrical part and a mechanical part. We first derive the electrical energy of the synchronous machine.

\subsubsection{Electrical energy}
In this section we search for an expression for the electrical energy of the reduced order models for the synchronous machine. A natural starting point is to look at the electrical energy of the full-order system and rewrite this in terms of the state variables of the reduced order system. Recall that the electrical energy in the $d$- and $q$-axis of the full-order system is respectively given by\footnote{For notational convenience the subscript $i$ is omitted in this section.} 
\begin{align*}
H_{dq}= H_d+H_q&= \F\begin{bmatrix}
    \Psi_d\\
    \Psi_f\\
    \Psi_{D}
  \end{bmatrix}^T
  \begin{bmatrix}
    L_d&kM_f&kM_D\\
    kM_f&L_{f}&L_{fD}\\
    kM_D&L_{fD}&L_{D}
  \end{bmatrix}^{-1}
  \begin{bmatrix}
    \Psi_d\\
    \Psi_f\\
    \Psi_{D}
  \end{bmatrix}\\&+    \F\begin{bmatrix}
    \Psi_q\\
    \Psi_g\\
    \Psi_{Q}
  \end{bmatrix}^T
  \begin{bmatrix}
    L_q&kM_g&kM_{Q}\\
    kM_g&L_g&L_{gQ}    \\
    kM_{Q}&L_{gQ}&L_{Q}
  \end{bmatrix}^{-1}
\begin{bmatrix}
    \Psi_q\\
    \Psi_g\\
    \Psi_{Q}
  \end{bmatrix}.
\end{align*}
Using the definitions of $E_q',E_q''$ and the reactances $X_d,X_d',X_{d}''$ we can, after involved rewriting\footnote{To obtain \eqref{eq:Hdpsid}  requires not only computing the inverse of the inductance matrices but also to appropriately eliminate the appropriate parameters and variables used in the model \eqref{eq:Mhfullordermodel}. Our calculations have been verified by computer algebra program Mathematica 11 $\textregistered$.}, 
 express the electrical energy in the $d$-axis as
\begin{align}\label{eq:Hdpsid}
H_{d}  &= \frac12 \begin{bmatrix}
    \Psi_d\\E_q'\\E_q''
  \end{bmatrix}^T\left [
\begin{array}{ccccc}
 \frac{\w_s }{X_d''} & 0                                                                        & -\frac{1}{X_d''}                                    \\
 0                      & \frac{1}{\w_s(X_d-X_d')}+\frac{1}{\w_s(X_d'-X_d'')} & -\frac{1}{\w_s  \left(X_d'-X_d''\right)}         \\
 -\frac{1}{X_d''}       & -\frac{1}{\w_s  \left(X_d'-X_d''\right)}                              & \frac{X_d'}{\w_s  \left(X_d'-X_d''\right) X_d''} \\
\end{array}
\right]
  \begin{bmatrix}
    \Psi_d\\E_q'\\E_q''
  \end{bmatrix}
\end{align}
and a similar expression for the energy $H_{q}$ can be derived for the $q$-axis. 

\paragraph{Sixth-order model}
We can also express the electrical energy \eqref{eq:Hdpsid} in term of the currents $I_d,I_q$ as follows. First, by Assumption \ref{ass:ssemfs} we eliminate $\Psi_d,\Psi_q$ by substituting $\Psi_q=-\w_s^{-1}(V_d+RI_d) ,\Psi_d=\w_s^{-1}(V_q+RI_q)$. 
Then $V_d,V_q$ can be eliminated by substituting \eqref{eq:VE}, that is,  $V_d=E_d''-RI_d-X_q''I_q, V_q=E_q''-RI_q+X_d''I_d$. Consequently,  for the sixth-order model the electrical energy stored in the machine takes the form
\begin{align}\label{eq:ensubtransq}
H_{d}  &=\frac{1}{2\w_s}  \begin{bmatrix}
    I_d\\E_q'\\E_q''
  \end{bmatrix}^T\left [
\begin{array}{ccccc}
  X_d'' & 0                                                               & 0                           \\
  0                & \frac{1}{X_d-X_d'}+\frac{1}{X_d'-X_d''} & -\frac{1}{X_d'-X_d''}                       \\
  0 & -\frac{1}{X_d'-X_d''}                                            & \frac{1}{X_d'-X_d''} \\
\end{array}
\right]
  \begin{bmatrix}
    I_d\\E_q'\\E_q''
  \end{bmatrix},
\end{align}
and a similar expression is obtained for the $q$-axis by exchanging the $dq$-subscripts. Remarkably, this is exactly the energy stored in the generator equivalent circuits illustrated in Figure \ref{fig:gencirc} in case $E_f=0$. 

\paragraph{Fifth-order model}
For the fifth-order model we have that $E_{d}'=0$ implying that the electrical energy in the $q$-axis modifies to 
\begin{align}\label{eq:ensubtransdEd0}
H_{q}  &=\frac{1}{2\w_s}  \begin{bmatrix}
    I_q\\E_d''
  \end{bmatrix}^T\left [
\begin{array}{ccccc}
  X_q'' &  0                     \\
  0     &  \frac{1}{X_q'-X_q''}  \\
\end{array}
\right]
  \begin{bmatrix}
    I_q\\E_d''
  \end{bmatrix},
\end{align}
while the expression for $H_d$ remains identical to the one for the sixth-order model, see equation \eqref{eq:ensubtransq}.

\paragraph{Lower-order models}
Since for the fourth, third and second-order model the subtransient dynamics is neglected, we can substitute \eqref{eq:subtrans0} into \eqref{eq:ensubtransq} such that the electrical energy $H_{dq}:=H_d+H_q$ can be written as 
\begin{align}\label{eq:entransq}
H_{dq}  &=\frac{1}{2\w_s}  \begin{bmatrix}
    I_d\\E_q'
  \end{bmatrix}^T\left [
\begin{array}{ccccc}
  X_d' & 0                \\
  0     & \frac1{X_d-X_d'} 
\end{array}
\right]
  \begin{bmatrix}
    I_d\\E_q'
  \end{bmatrix}
+\frac{1}{2\w_s}  \begin{bmatrix}
    I_q\\E_d'
  \end{bmatrix}^T\left [
\begin{array}{ccccc}
  X_q' & 0                \\
  0     & \frac1{X_q-X_q'} 
\end{array}
\right]
  \begin{bmatrix}
    I_q\\E_d'
  \end{bmatrix}
\end{align}
and for the third-order model we have $E_d'=0$. 
\begin{remark}[Synchronous machines reactances as part of line reactances]\label{rem:Xdpartnetwork}
  If the (sub)transient saliency is neglected then the reactance $X_d' \ (X_d'')$ can considered as part the (transmission) network, see Section \ref{sec:multi-machine-models}. Therefore, the energy stored in this reactance will be part of the energy stored in the transmission lines which will be discussed in Section \ref{sec:transmission-lines}. As a result, the part of the energy \eqref{eq:entransq} corresponding with $I_d,I_q$ can be disregarded here. For example, for the \text{fourth-,} third- and second-order model the energy function associated to the electrical energy stored in the generator circuit is given by
  \begin{align}\label{eq:Hdq234o} H_{dq}&=\frac1{2\w_s}\frac{(E_{q}')^2}{X_{d}-X_{d}'}+\frac1{2\w_s}\frac{(E_{d}')^2}{X_{q}-X_{q}'},
  \end{align}
where $E_d'=0$ for the third-order model. 
\end{remark}
Bearing in mind Remark \ref{rem:Xdpartnetwork} and noting that for the second-order model the voltages $E_q',E_d'$ are constant, it follows that the electrical energy \eqref{eq:Hdq234o} is constant as well. 

\subsubsection{Mechanical energy}
The rotational kinetic energy of synchronous machine $i$ is given by
\begin{align}\label{eq:Hmi}
H_{mi}=\frac12J_i\w_i^2=\frac1{2}J_i^{-1}p_i^2
\end{align}
where we recall that the angular momentum is defined by $p_i=J_i\w_i$. 

\subsection{Inductive transmission lines}
\label{sec:transmission-lines}

\subsubsection{Sixth- and fifth-order models}
Consider an inductive transmission line between nodes $i$ and $k$ at steady state, see  Figure \ref{fig:indline}.
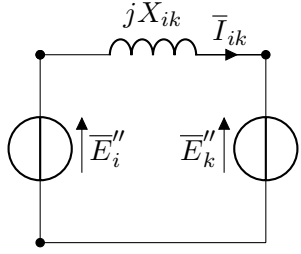
\begin{figure}
\begin{center}
\begin{circuitikz}
\draw[]  (0,0)   to [european voltage source,*-*,v_=$\overline E_{i}''$]  ++(0,2.5)
 to [american inductor,l^=$ jX_{ik} $,*-*,i=$\overline I_{ik}$
 ] ++(3,0)
 to [european voltage source,v_<=$\overline E_{k}''$] ++(0,-2.5)  
to ++(-3,0) 
;
\end{circuitikz}
\end{center}
\caption{An inductive transmission line at steady state. The internal voltages $\overline E_{i}'', \overline E_{k}''$ are expressed in the corresponding local $dq0$-reference frame.} 
\label{fig:indline}
\end{figure}
When expressed in the local $dq$-reference frame of synchronous machine $i$, we observe  from  Figure \ref{fig:indline} that 
  \begin{align}\label{eq:56-current-voltage}
    jX_{ik}\overline I_{ik}=\overline{E}_i''-e^{-j\delta_{ik}}\overline{E}_k''.
  \end{align}
By equating the real and imaginary part of \eqref{eq:56-current-voltage} we obtain
\begin{align}\label{eq:IVindline}
 X_{ik} \begin{bmatrix}
    I_{qik}\\
    -I_{dik}
  \end{bmatrix}&=
     \begin{bmatrix}
   E_{di}''-E_{dk}'' \cos \delta_{ik}+E_{qk}'' \sin \delta_{ik}\\
   E_{qi}''-E_{dk}'' \sin \delta_{ik}-E_{qk}''\cos \delta_{ik}
  \end{bmatrix}.
\end{align}
Note that the energy of the inductive transmission line between nodes $i$ and $k$ is given by 
$$H_{ik}=\frac12L_{ik}\overline I_{ik}^*\overline I_{ik}=\frac{X_{ik}}{2\w_s}(I_{dik}^2+I_{qik}^2)$$ 
which by \eqref{eq:IVindline}  can be written as 
\begin{equation}\label{eq:EnIndTransLines}
\begin{aligned}
  H_{ik}
&=-\frac{B_{ik}}{\w_s}\Big[\left(E_{di}'' E_{qk}''-E_{dk}'' E_{qi}''\right)\sin \delta _{ik}-  \left(E_{di}'' E_{dk}''+E_{qi}'' E_{qk}''\right)\cos \delta _{ik}\\&+\tfrac12E_{di}''^2+\tfrac12E_{dk}''^2+\tfrac12E_{qi}''^2+\tfrac12E_{qk}''^2\Big]
\end{aligned}
\end{equation}
where $B_{ik}=-\frac1{X_{ik}}<0$ is the susceptance of transmission line $(i,k)$ \cite{schiffer2016survey}.  

\subsubsection{Fourth- and third-order models}
For the fourth- and third-order model the transient reactances\footnote{Provided that the transient saliency is neglected, i.e., $X_{di}'=X_{qi}'$ for all $i\in\V$.} $X_{di}'$ can be considered as part of the network implying that the energy in the transmission lines can be obtained by replacing the subtransient voltages by the transient voltages in \eqref{eq:EnIndTransLines}. For the third-order model $E_{di}'=0$ for all $i\in\V$ so that the energy function associated to the transmission line between node $i$ and $k$ simplifies to 
\begin{equation}\label{eq:EnIndTransLines3order}
\begin{aligned}
  H_{ik}&=-\frac{B_{ik}}{\w_s}\left(\tfrac12E_{qi}'^2+\tfrac12E_{qk}'^2- E_{qi}' E_{qk}' \cos \delta _{ik}\right).
\end{aligned}
\end{equation}

\subsubsection{Second-order model}
For the second-order model it is convenient to represent transient voltages as $\overline E_i'=|\overline E_i'|e^{j\alpha_i}$ where $\alpha_i$ is the voltage angle of $\overline E_i'$ with respect to the rotor angle. Then, by defining the voltages angles $\theta_i=\delta_i+\alpha_i$ as in Section \ref{sec:class-multi-mach}, the energy in the transmission line\footnote{Where the subtransient voltages are replaced by the transient voltages.} \eqref{eq:EnIndTransLines} takes the much simpler form
\begin{equation}\label{eq:EnIndTransLines2order}
\begin{aligned}
  H_{ik} & =-\frac{B_{ik}}{2\w_s}(\overline E_i'-\overline E_k'e^{-j\delta_{ik}})^*(\overline E_i'-\overline E_k'e^{-j\delta_{ik}})\\
         & =-\frac{B_{ik}}{2\w_s}(|\overline E_i'|^2-2|\overline E_i'||\overline E_k'|\cos\theta_{ik}+|\overline E_k'|^2).
\end{aligned}
\end{equation}

\subsection{Total energy}
The total energy of the multi-machine system is equal to the sum of the previously mentioned energy functions 
\begin{align}\label{eq:Htotal}
  H=\sum_{i\in\V}\left(H_{di}+H_{qi}+H_{mi}\right)+\sum_{(i,k)\in\mathcal E}H_{ik},
\end{align}
where the expressions for each individual energy function depends on the order of the model. The resulting energy function $H$  could serve as a candidate Lyaponuv function for the stability analysis of the multi-machine power network (with zero inputs). 


\begin{remark}[Common factor $\w_s^{-1}$ in energy function]
It is observed that each of the individual energy functions appearing in \eqref{eq:Htotal} contains a factor $\w_s^{-1}$. Therefore, a modified version of the energy function defined by $U=\w_sH$ 
can also be used as a Lyapunov function for the multi-machine system. However, the function $U$ does not have the dimension of energy anymore, but has the dimension of power instead. In fact, in most of the literature these modified energy functions\footnote{Which are sometimes incorrectly called \emph{energy functions} as well.} (without the factor $\w_s^{-1}$) are (part of) the collection of Lyapunov functions used to analyze the stability of the power network, see e.g. \cite{LHMNLC,trip2016internal,zhangpapaautomatica,zhangpapa,AGC_ACC2014,pai1989energy}.
\end{remark}


\section{Port-Hamiltonian framework}
\label{sec:port-hamilt-repr}
By using the energy function established in the previous section, a convenient representation of the multi-machine models of Section \ref{sec:multi-machine-models} can be obtained. This is based on the theory of \pH\ systems, which yields a systematic framework for network modelling of multi-physics systems. In particular, we show in this section that the complex multi-machine systems \eqref{eq:multimach6},~\eqref{eq:multimach3},~\eqref{eq:SG2multimach} admit a simple \pH\ representation. 
Finally, some important passivity properties are proven for the resulting systems. 

\subsection{Sixth-order model}

\subsubsection{Energy in the transmission lines}
Recall from \eqref{eq:EnIndTransLines} that the energy stored in the inductive transmission line between node $i$ and $k$ is given by
\begin{equation}\label{eq:EnIndTransLines2}
\begin{aligned}
  H_{ik}
&=-\frac{B_{ik}}{\w_s}\Big[\left(E_{di}'' E_{qk}''-E_{dk}'' E_{qi}''\right)\sin \delta _{ik}-  \left(E_{di}'' E_{dk}''+E_{qi}'' E_{qk}''\right)\cos \delta _{ik}\\&+\tfrac12E_{di}''^2+\tfrac12E_{dk}''^2+\tfrac12E_{qi}''^2+\tfrac12E_{qk}''^2\Big]
\end{aligned}
\end{equation}
where $B_{ik}=-\frac1{X_{ik}}<0$ according to \cite{schiffer2016survey}. Observe that the gradient of $H_{ik}$ takes the form 
\begin{align*}
  \begin{bmatrix}
    \frac{\p H_{ik}}{\p \delta_i}\\
    \frac{\p H_{ik}}{\p E_{qi}''}\\
    \frac{\p H_{ik}}{\p E_{di}''}
  \end{bmatrix}=\frac{B_{ik}}{\w_s}
  \begin{bmatrix}
    (E_{qi}''E_{dk}''-E_{di}''E_{qk}'')\cos \delta_{ik}-(E_{di}'' E_{dk}''+E_{qi}'' E_{qk}'')\sin \delta_{ik}\\
-E_{qi}''+E_{qk}'' \cos \delta_{ik}-E_{dk}''\sin\delta_{ik}\\
-E_{di}''+E_{dk}''\cos\delta_{ik}+E_{qk}''\sin\delta_{ik}
  \end{bmatrix}.
\end{align*}
After defining the total energy stored in the transmission lines by  $H_{T}=\sum_{(i,k)\in\E}H_{ik}$, we obtain likewise
\begin{align*}
  \begin{bmatrix}
    \frac{\p H_{T}}{\p \delta_i}\\
    \frac{\p H_{T}}{\p E_{qi}''}\\
    \frac{\p H_{T}}{\p E_{di}''}
  \end{bmatrix}
&=\frac1{\w_s}
  \begin{bmatrix}
   \sum_{k\in\mathcal N_i}B_{ik}[ (E_{qi}''E_{dk}''-E_{di}''E_{qk}'')\cos \delta_{ik}-(E_{di}'' E_{dk}''+E_{qi}'' E_{qk}'')\sin \delta_{ik}]\\
-B_{ii}E_{qi}''+\sum_{k\in\mathcal N_i}B_{ik}(E_{qk}'' \cos \delta_{ik}+E_{dk}''\sin\delta_{ik})\\
-B_{ii}E_{di}''+\sum_{k\in\mathcal N_i}B_{ik}(E_{dk}''\cos\delta_{ik}+E_{qk}''\sin\delta_{ik})
  \end{bmatrix}\\&=\frac1{\w_s}
    \begin{bmatrix}
      P_{ei}\\
      -I_{di}\\
      I_{qi}
    \end{bmatrix}
\end{align*}
where we have used the fact that $B_{ii}=\sum_{k\in\mathcal N_i}B_{ik}$ and  equations  \eqref{eq:currentsmulti}, \eqref{eq:Pei6order}. 

\subsubsection{Electrical energy synchronous machine}
Further notice that the electrical energy stored in the $d$-axis in machine $i$  is given by
\begin{align*}
H_{di}=\frac1{2\w_s}
  \begin{bmatrix}
    E_{qi}'&
    E_{qi}''
  \end{bmatrix}
             \begin{bmatrix}
               \frac{1}{X_{di}-X_{di}'}+\frac{1}{X_{di}'-X_{di}''}&-\frac{1}{X_{di}'-X_{di}''}\\
-\frac{1}{X_{di}'-X_{di}''}&\frac{1}{X_{di}'-X_{di}''}
             \end{bmatrix}
  \begin{bmatrix}
    E_{qi}'\\
    E_{qi}''
  \end{bmatrix}
\end{align*}
and satisfies
\begin{align*}
  \begin{bmatrix}
    X_{di}-X_{di}'& X_{di}-X_{di}'
  \end{bmatrix}
  \begin{bmatrix}
    \frac{\p H_{di}}{\p E_{qi}'}\\
    \frac{\p H_{di}}{\p E_{qi}''}
  \end{bmatrix}&=\frac1{\w_s}E_{qi}'\\
  \begin{bmatrix}
    0& X_{di}'-X_{di}''
  \end{bmatrix}
  \begin{bmatrix}
    \frac{\p H_{di}}{\p E_{qi}'}\\
    \frac{\p H_{di}}{\p E_{qi}''}
  \end{bmatrix}&=\frac1{\w_s}(E_{qi}''-E_{qi}').
\end{align*}
Observe that a  similar result can be established for the energy function  $H_{qi}$ by exchanging the $d$- and $q$-subscripts. 

\subsubsection{Mechanical energy}
\label{sec:mechanical-energy}
To obtain a \pH\ representation of the multi-machine models, it is convenient to rephrase and shift the energy function \eqref{eq:Hmi}  with respect to the synchronous frequency to obtain
\begin{align*}
\bar H_{mi}=\frac12J_i\Delta\w_i^2=\frac1{2\w_s}M_i\Delta\w_i^2=\frac1{2\w_s}M_i^{-1}\mathtt{p}_i^2,
\end{align*}
where $M_i=\w_sJ_i$ and we define the variable $\mathtt p_i=M_i\Delta\w_i$.
\begin{remark}[Modified 'moment of inertia']
  Note that the quantity $\mathtt p_i$ does not represent the angular momentum of the synchronous machine but instead it is equal to $\mathtt p_i=\w_sJ\Delta\w_i$ so it has a different physical dimension. In addition,  it is shifted with respect to the synchronous frequency. 
\end{remark}
Using this definition of the Hamiltonian  $\bar H_{mi}(\mathtt{p}_i)$, it follows that its gradient satisfies
\begin{align*}
  \frac{\p \bar H_{mi}}{\p \mathtt p_i}(\mathtt{p}_i)=\frac1{\w_s}M_i^{-1}\mathtt{p}_i=\frac{\Delta \w_i}{\w_s}.
\end{align*}

\subsubsection{Port-Hamiltonian representation}
By the previous observations, the dynamics of a single synchronous machine in a multi-machine system \eqref{eq:multimach6} can be written in the form
\begin{equation}\label{eq:pH6multiphatE}
\begin{aligned}
\begin{bmatrix}
\dot  \delta_i\\
  \dot {\mathtt p}_i\\
\dot  E_{qi}'\\
\dot  E_{di}'\\
\dot E_{qi}''\\
\dot E_{di}''
\end{bmatrix}&=\w_s
  \begin{bmatrix}
   0 &  1  & 0                             & 0                             & 0                               & 0                             \\
   -1&  0  & 0                             & 0                             & 0                               & 0                             \\
   0 &  0  & -\frac{\hat X_{di}}{T_{doi}'} & 0                             & -\frac{\hat X_{di}}{T_{doi}'}   & 0                             \\
   0 &  0  & 0                             & -\frac{\hat X_{qi}}{T_{qoi}'} & 0                               & -\frac{\hat X_{qi}}{T_{qoi}'} \\
   0 &  0  & 0                             & 0                             & -\frac{\hat X_{di}'}{T_{doi}''} & 0                             \\
   0 &  0  & 0                             & 0                             & 0                               & -\frac{\hat X_{qi}'}{T_{qoi}''}
  \end{bmatrix}
\nabla_i H +
\begin{bmatrix}
0&0\\
1&0\\
0&\frac1{T_{doi}'}\\    
0&0\\
0&0\\
0&0
  \end{bmatrix}
  \begin{bmatrix}
    P_{mi}\\
    E_{fi}
  \end{bmatrix}
\\
y_i&=
   \begin{bmatrix}
     0&1&0&0&0&0\\
     0&0&\frac1{T_{doi}'}&0&0&0
   \end{bmatrix}
\nabla_iH
\end{aligned}
\end{equation}
where 
\begin{align*} H=\sum_{i\in\V}\Big(\bar H_{mi}+H_{di}+H_{qi} \Big)+\sum_{(i,k)\in\E}H_{ik}
\end{align*}
and $\hat X_{di}:=X_{di}-X_{di}',\hat X_{di}':=X_{di}'-X_{di}'',\hat X_{qi}:=X_{qi}-X_{qi}',\hat X_{qi}':=X_{qi}'-X_{qi}''$ and $\nabla_iH$ denotes the gradient of $H$ with respect to the variables $\col(\delta_i,\mathtt{p}_i,  E_{qi}' , E_{di}', E_{qi}'', E_{di}'')$.
Note that the mechanical energy $\bar H_{mi}$ is \emph{shifted} around the synchronous frequency. 
By aggregating the states of the synchronous machines, i.e. $\delta=\col(\delta_1,\ldots,\delta_n)$ etc., the multi-machine  system is described by 
\begin{equation}
\begin{aligned}
\begin{bmatrix}
\dot   \delta\\
  \dot{\mathtt p}\\
\dot   E_{q}'\\
\dot   E_{d}'\\
\dot E_{q}''\\
\dot E_{d}''
\end{bmatrix}&=\underbrace{\w_s
  \begin{bmatrix}
   0  & I &  0                         & 0                         & 0                           & 0                         \\
   -I & 0 &  0                         & 0                         & 0                           & 0                         \\
   0  & 0 &  -(T_{do}')^{-1}\hat X_{d} & 0                         & -(T_{do}')^{-1}\hat X_{d}   & 0                         \\
   0  & 0 &  0                         & -(T_{qo}')^{-1}\hat X_{q} & 0                           & -(T_{qo}')^{-1}\hat X_{q} \\
   0  & 0 &  0                         & 0                         & -(T_{do}'')^{-1}\hat X_{d}' & 0                         \\
   0  & 0 &  0                         & 0                         & 0                           & -(T_{qo}'')^{-1}\hat X_{q}'
  \end{bmatrix}}_{J-R}
\nabla H          \\&+
g
  \begin{bmatrix}
    P_{m}\\
    E_{f}
  \end{bmatrix}, \qquad y=g^T
\nabla H, \qquad g=  \begin{bmatrix}
     0&I&0&0&0&0\\
     0&0&(T_{do}')^{-1}&0&0&0
   \end{bmatrix}^T,
\end{aligned}\label{eq:6multiphall}
\end{equation}
where $\hat X_{d}=\diag(\hat X_{d1},\ldots,\hat X_{dn}), T_{do}'=\diag(T_{do1}',\ldots,T_{don}')$ and likewise definitions are used for the quantities $\hat X_{d}',\hat X_{q},\hat X_q',T_{qo}',T_{do}'',T_{qo}''$. The matrix $J-R$ depicted in equation \eqref{eq:6multiphall} consists of a skew-symmetric matrix $J=-J^T$ and a symmetric matrix $R=R^T$ often called the \emph{dissipation matrix} \cite{phsurvey}. Provided that dissipation matrix is positive semi-definite, i.e. $R\geq0$, the system  \eqref{eq:6multiphall} is indeed a \pH\ representation of the sixth-order multi-machine network \eqref{eq:multimach6}.
\begin{proposition}[Positive semi-definite dissipation matrix]\label{prop:dissmatr}
  Suppose that the following holds true:
  \begin{subequations}\label{eq:neccTdq}
\begin{align}
  4(X_{di}'-X_{di}'')T_{doi}'-(X_{di}-X_{di}')T_{doi}''&\geq 0,\label{eq:neccTd}\\
  4(X_{qi}'-X_{qi}'')T_{qoi}'-(X_{qi}-X_{qi}')T_{qoi}''&\geq 0,\label{eq:neccTq}
\end{align}
\end{subequations}
for all $i\in\V$. Then \eqref{eq:6multiphall} is a \pH\ representation of the 6-order  multi-machine network \eqref{eq:multimach6}. 
\end{proposition}
\begin{proof}
  The dissipation matrix of the  system \eqref{eq:6multiphall} is equal to the symmetric part of the matrix in \eqref{eq:6multiphall} and amounts to
  \begin{align*}
R= \w_s \begin{bmatrix}
    0&0&0&0&0&0\\
    0&0&0&0&0&0\\
    0&0&-(T_{do}')^{-1}\hat X_{d}&0&-\frac12(T_{do}')^{-1}\hat X_{d}&0\\
    0&0&0&-(T_{qo}')^{-1}\hat X_{q}&0&-\frac12(T_{qo}')^{-1}\hat X_{q}\\
    0&0&-\frac12(T_{do}')^{-1}\hat X_{d}&0&-(T_{do}'')^{-1}\hat X_{d}'&0\\
    0&0&0&-\frac12(T_{qo}')^{-1}\hat X_{q}&0&-(T_{qo}'')^{-1}\hat X_{q}'
  \end{bmatrix}.
  \end{align*}
By invoking the Schur complement, $R=R^T\geq0$  if and only if 
\begin{align*}
\begin{bmatrix}
  2\frac{X_{di}-X_{di}'}{T_{doi}'}&\frac{X_{di}-X_{di}'}{T_{doi}'}\\
\frac{X_{di}-X_{di}'}{T_{doi}'}&2\frac{X_{di}'-X_{di}''}{T_{doi}''}
\end{bmatrix}\geq0 \quad  \text{and} \quad 
\begin{bmatrix}
  2\frac{X_{qi}-X_{qi}'}{T_{qoi}'}&\frac{X_{qi}-X_{qi}'}{T_{qoi}'}\\
\frac{X_{qi}-X_{qi}'}{T_{qoi}'}&2\frac{X_{qi}'-X_{qi}''}{T_{qoi}''}
\end{bmatrix}\geq0, \quad  \forall i\in\V,
\end{align*}
which holds if and only if \eqref{eq:neccTdq} holds for all $i\in\V.$ 
\end{proof}
\begin{remark}[Condition \eqref{eq:neccTdq} holds strictly]
  It should be stressed that for a typical generator $T_{doi}''\ll T_{doi}'$ and $T_{qoi}''\ll T_{qoi}'$ implying that \eqref{eq:neccTdq}   holds (in the strict sense), see also Table 4.3 of \cite{powsysdynwiley} or Table 4.2 of \cite{kundur}. 
  In fact, we claim that 
  \begin{align*}
    &4(X_{di}'-X_{di}'')T_{doi}'-(X_{di}-X_{di}')T_{doi}''\\
    &\quad =
  \kappa^2\w_s\cdot\frac{4 L_f^2 (L_f M_D - L_{fD} M_f)^2 R_D + (L_D L_f - 
      L_{fD}^2)^2 M_f^2 R_f}{L_f^2 (L_D L_f - L_{fD}^2) R_D R_f}>0\\
      &4(X_{qi}'-X_{qi}'')T_{qoi}'-(X_{qi}-X_{qi}')T_{qoi}''\\
  &\quad=\kappa^2\w_s\cdot\frac{4 L_g^2 (L_g M_Q - L_{gQ} M_g)^2 R_Q + (L_Q L_g - 
      L_{gQ}^2)^2 M_g^2 R_g}{L_g^2 (L_Q L_g - L_{gQ}^2) R_Q R_g}>0
  \end{align*}
The equalities are verified 
by substituting the quantities from \eqref{eq:subtransientindtime}. In addition, the inequalities hold since $R_D,R_Q,R_f,R_g,L_f,L_g,L_D,L_Q>0$ and, since $X_d-X_d'>0,X_d-X_d'>0$, we have that $M_f\neq 0,M_g\neq 0$. Finally, $L_DL_f-L_{fD}^2>0,L_QL_g-L_{fQ}^2>0$ as the inductance matrices $\mathcal L_d,\mathcal L_q$ defined in \eqref{eq:Md-ind},~\eqref{eq:Mhq-ind} are positive definite. Hence, it is emphasized that \eqref{eq:neccTdq} holds in the strict sense for a realistic synchronous machine.
\end{remark}

\begin{remark}[Fifth and fourth order models]
  It can be shown that similar \pH\ structures appear for the fifth- and fourth-order multi-machine networks using the corresponding (shifted) energy functions derived in Section \ref{sec:energy-analysis} and Section \ref{sec:mechanical-energy} as the Hamiltonian. 
\end{remark}

\subsection{Third-order model}

Recall from \eqref{eq:EnIndTransLines3order} that the energy stored in the inductive transmission line between node $i$ and $k$ is given by
\begin{equation}\label{eq:EnIndTransLines3order2}
\begin{aligned}
  H_{ik}&=-\frac{B_{ik}}{\w_s}\left(\tfrac12E_{qi}'^2+\tfrac12E_{qk}'^2- E_{qi}' E_{qk}' \cos \delta _{ik}\right).
\end{aligned}
\end{equation}
Observe that the gradient of $H_{ik}$ is given by 
\begin{align*}
  \begin{bmatrix}
    \frac{\p H_{ik}}{\p \delta_i}\\
    \frac{\p H_{ik}}{\p E_{qi}'}
  \end{bmatrix}=\frac{B_{ik}}{\w_s}
  \begin{bmatrix}
    -E_{qi}' E_{qk}'\sin \delta_{ik}\\
-E_{qi}'+E_{qk}' \cos \delta_{ik}
  \end{bmatrix}.
\end{align*}
Define now the total energy stored in the transmission lines by  $H_{T}=\sum_{(i,k)\in\E}H_{ik}$. Then we obtain likewise
\begin{align*}
  \begin{bmatrix}
    \frac{\p H_{T}}{\p \delta_i}\\
    \frac{\p H_{T}}{\p E_{qi}'}
  \end{bmatrix}
&=\frac1{\w_s}
  \begin{bmatrix}
   \sum_{k\in\mathcal N_i}B_{ik} -E_{qi}' E_{qk}'\sin \delta_{ik}\\
-B_{ii}E_{qi}'+\sum_{k\in\mathcal N_i}B_{ik}E_{qk}' \cos \delta_{ik}
  \end{bmatrix}=\frac1{\w_s}
    \begin{bmatrix}
      P_{ei}\\
      -I_{di}
    \end{bmatrix}.
\end{align*}
Further notice that the electrical energy stored in machine $i$ is given by
\begin{align*}
H_{dqi}=\frac1{2\w_s}\frac{(E_{qi}')^2}{X_{di}-X_{di}'}
\end{align*}
and satisfies
\begin{align*}
   (X_{di}-X_{di}')\frac{\p H_{dqi}}{\p E_{qi}'}&=\frac1{\w_s}E_{qi}'.
\end{align*}
By the previous observations and aggregating the states, the dynamics of the third-order multi-machine system \eqref{eq:multimach3} can now be written in \pH\ form  as
\begin{equation}
\begin{aligned}
\begin{bmatrix}
\dot   \delta\\
  \dot{\mathtt p}\\
\dot   E_{q}'
\end{bmatrix}&=\w_s
  \begin{bmatrix}
     0  & I  &0 \\
     -I & -D &0 \\
     0  & 0  &-(T_{do}')^{-1}(X_d-X_d')
  \end{bmatrix}
\nabla H+g
  \begin{bmatrix}
    P_{m}\\
    E_{f}
  \end{bmatrix}, \\y&=g^T
\nabla H, \quad g^T=  \begin{bmatrix}
    0& I&0\\
    0& 0&(T_{do}')^{-1}
   \end{bmatrix}, \quad H=\sum_{i\in\V}\big(\bar H_{mi}+H_{dqi}\big)+\sum_{(i,k)\in\E}H_{ik}
\end{aligned}\label{eq:3multiphall}
\end{equation}
where $X_{d}=\diag ( X_{d1},\ldots,X_{dn}), X_{d}'=\diag ( X_{d1}',\ldots,X_{dn}'),$ and  in addition $ T_{do}'=\diag(T_{do1}',\ldots,T_{don}')$.  


\subsection{Swing equations}
Recall from \eqref{eq:EnIndTransLines2order} that the energy stored in the inductive transmission line between node $i$ and $k$ is given by
\begin{equation}\label{eq:EnIndTransLines2order2}
\begin{aligned}
  H_{ik}&=         & =-\frac{B_{ik}}{2\w_s}(|\overline E_i'|^2-2|\overline E_i'||\overline E_k'|\cos\theta_{ik}+|\overline E_k'|^2).
\end{aligned}
\end{equation}
Define now the total energy stored in the transmission lines by  $H_{T}=\sum_{(i,k)\in\E}H_{ik}$ and observe that the gradient of $H_{T}$ with respect to the transformed angle $\theta$ is given by 
\begin{align*}
    \frac{\p H_{T}}{\p \theta_i}=-\sum_{j\in\mathcal N_i}\frac{B_{ik}}{\w_s}|\overline E_{i}'|| \overline E_{k}'|\sin \theta_{ik}.
\end{align*}
For the second-order model the electrical energy stored in the generator circuits is constant and can therefore be omitted from the Hamiltonian without loss of generality. By the previous observations and aggregating the states, the dynamics of the second-order multi-machine system \eqref{eq:SG2multimach} with $G=0$ can be written in \pH\ form as
\begin{equation}
\begin{aligned}
\begin{bmatrix}
\dot   \delta\\
  \dot{\mathtt p}
\end{bmatrix}&=\w_s
  \begin{bmatrix}
    0 & I  \\
    -I& -D   
  \end{bmatrix}
\nabla H+
\begin{bmatrix}
  0\\I
\end{bmatrix}
    P_{m}
 \\y&=
 \begin{bmatrix}
   0&I
 \end{bmatrix}
\nabla H=\frac{\Delta \w}{\w_s}, \qquad H=\sum_{i\in\V}\bar H_{mi}+\sum_{(i,k)\in\E}H_{ik}.
\end{aligned}\label{eq:2multiphall}
\end{equation}

\subsection{Passivity}
Since the multi-machine systems  \eqref{eq:6multiphall}, \eqref{eq:3multiphall}, \eqref{eq:2multiphall} are in the \pH\ form 
\begin{equation}
\begin{aligned}
  \dot x&=(J-R)\nabla H(x)+g u\\
  y&=g^T\nabla H(x)
\end{aligned}\label{eq:consJRpHsys}
\end{equation}
with constant matrices $J=-J^T,R=R^T\geq0$, they satisfy the following shifted passivity property.   
\begin{proposition}[Shifted passivity]\label{prop:shpass}
  Let $\bar u$ be a constant input and suppose there exists a corresponding equilibrium $\bar x$ to \eqref{eq:consJRpHsys} that satisfies $\nabla^2 H(\bar x)>0$.  Then the system \eqref{eq:consJRpHsys} is passive w.r.t. the shifted external port-variables $\tilde u:=u-\bar u, \tilde y:=y-\bar y$ with  $\bar y=g^T\nabla H(\bar x)$, and the local storage function is given by $\bar H(x):=H(x)-(x-\bar x)^T\nabla H(\bar x)-H(\bar x)$.
\end{proposition}
\begin{proof}
  By defining the \emph{shifted Hamiltonian} (see e.g. \cite{phsurvey}) as $\bar H(x):=H(x)-(x-\bar x)^T\nabla H(\bar x)-H(\bar x)$ as in the proposition,  the system \eqref{eq:consJRpHsys} can be rewritten as 
  \begin{align*}
    \dot x&=(J-R)\nabla H(x)+gu=(J-R)(\nabla \bar H(x)+\nabla H(\bar x))+gu\\
&=(J-R)\nabla \bar H(x)+g(u-\bar u)=(J-R)\nabla \bar H(x)+g\tilde u\\
\tilde y&=y-\bar y=g^T(\nabla H(x)-\nabla H(\bar x))=g^T\nabla \bar H(x).
  \end{align*}
The passivity follows by taking the time-derivative of the shifted Hamiltonian $\bar H$ which yields 
\begin{align*}
  \dot {\bar H}&=-(\nabla \bar H(x))^TR\nabla \bar H(x)+\tilde u^T\tilde y\leq \tilde u^T\tilde y.
\end{align*}
Since in addition $\nabla^2\bar H(\bar x)=\nabla^2 H(\bar x)>0$, it follows that $\bar H$ acts as a suitable local storage function.
\end{proof}
\begin{remark}[Hessian condition]
  To use Proposition \ref{prop:shpass} one must verify that the Hessian of the Hamiltonian evaluated at the (desired) equilibrium  is positive definite.  For the second and third-order multi-machine models a sufficient condition is established for guaranteeing that the Hessian is positive definite, see \cite{swing-claudio,de2018bregman}. It can be verified that these conditions hold for a typical operation point of the power network, i.e., for which the voltage (angle) differences are small. However, further effort is required to establish a similar condition for the higher-order multi-machine models, which preferably can be checked in a distributed fashion. 
\end{remark}
  The passivity property mentioned in Proposition \ref{prop:shpass} that the previously derived multi-machine models \eqref{eq:6multiphall}, \eqref{eq:3multiphall}, \eqref{eq:2multiphall} admit proves to be very useful when interconnection with (passive and optimal) controllers, see  in particular our previous work \cite{stegink2016optimal,stegink2017unifying} for an analysis of the third- and sixth-order models respectively. 

\section{Conclusions and future research}
\label{sec:conclusions}
In this paper a unifying energy-based approach to the modelling of multi-machine power networks is provided. Starting from the first-principle model of the synchronous generator, reduced order models are obtained and the underlying assumptions are explained. After determining the energy functions of the reduced-order models, a \pH\ representation of the multi-machine systems is established.
In particular, it is shown that advanced multi-machine models that are much more advanced can be analyzed using the \pH\ framework. 
Moreover, the resulting \pH\ system is proven to be shifted passive with respect to its steady states. The latter property has turned out to be crucial in many contexts, in particular for the stability analysis of the (optimal) equilibria of the closed-loop system \cite{stegink2017unifying,de2018bregman,trip2016internal}.

\subsection{Future research}
The results established in this paper can be extended in many possible ways. We elaborate on the main research directions in the following.

\subsubsection{Control}
One natural extension of the work established in the present paper is to consider (distributed) control of multi-machine networks. For frequency control, this can for example be done following the lines of \cite{stegink2016optimal,DAPIsimpson2013synchronization,trip2016internal}. 
Since in the present paper we established a systematic way for obtaining the energy functions and proved (shifted) passivity of the system, we conjecture that the same kind of controllers established in these references can be applied to (purely inductive) multi-machine models where each synchronous machine is described by a 2,3,4,5 or 6th-order model. In particular, the 6th-order multi-machine case was already been in our previous work \cite{stegink2016optimal}. 
Alternatively, one can continue along the lines of \cite{stegink2017unifying,LHMNLC,AGC_ACC2014,zhang1achieving} and consider controllers based on the primal-dual gradient method. 
In addition, further effort is required to investigate the possibilities of (optimal) voltage control using passive controllers. One possibility is to extend the work of \cite{de2016lyapunov,de2018bregman} to high-dimensional multi-machine models. 

\subsubsection{Nonzero transfer conductances}

Another extension to this work is to include transmission line resistances in the network. However, in \cite{ortega2005transient,bretas2003lyapunov} and references therein it is observed that in the case of nonzero transfer conductances, a Lyapunov based stability analysis can be cumbersome and involves adding nontrivial cross terms in the Lyapunov function. Even then, the stability analysis relies on a 'sufficiently small transfer conductances' assumption \cite{ortega2005transient,bretas2003lyapunov}.  
On the other hand, one approach that could be adopted in future research is to assume the resistive transmission lines are uniform such that the $R/X$ ratios are identical for all transmission lines. This simplifies the analysis and possibly the present work could be extended to this case (and keeping the \pH\ structure intact), for example by following the lines of \cite{de2018bregman} and references therein. 

\subsubsection{More accurate power network models}
In the present paper we considered the case that each node in the network represents a synchronous machine. A natural extension is to generalize the established results to the case where some of the nodes represent inverters or (frequency-dependent) loads instead. 
In addition, while advanced models of the synchronous generator are considered in this paper, there are many possible extensions to these models. For example,  models for the turbine and speed governor as considered in e.g. \cite{alv_meng_power_coupl_market,trip2017distributed,zhang1achieving} could also be taken into account. Finally, the model can be expanded such that the excitation system and the automatic voltage regulator (AVR) are included as well \cite{powsysdynwiley}.

\section*{Funding}
This work is supported by the Netherlands Organisation for Scientific Research (NWO) programme \emph{Uncertainty Reduction in Smart Energy Systems (URSES)} under the auspices of the project \emph{Energy-based analysis and control of the grid: dealing with uncertainty and markets (ENBARK)}.

\bibliographystyle{tfq}
\bibliography{mybib}

\begin{thebibliography}{10}
\newcommand{\printfirst}[2]{#1}
\newcommand{\switchargs}[2]{#2#1}
\providecommand{\url}[1]{\normalfont{#1}}
\providecommand{\urlprefix}{Available at }

\bibitem{phsurvey}
A.J. {van der Schaft} and D. Jeltsema, \emph{Port-{H}amiltonian systems theory:
  An introductory overview}, Foundations and Trends in Systems and Control 1
  (2014), pp. 173--378.

\bibitem{EJCFiaz}
S. Fiaz, D. Zonetti, R. Ortega, J.M.A. Scherpen, and A.J. {van der Schaft},
  \emph{A port-{H}amiltonian approach to power network modeling and analysis},
  European Journal of Control 19 (2013), pp. 477--485.

\bibitem{stegink2015port}
T.W. Stegink, C. {De Persis}, and A.J. {van der Schaft}, \emph{A
  port-{H}amiltonian approach to optimal frequency regulation in power grids},
  in \emph{54th IEEE Conference on Decision and Control (CDC)}. 2015, pp.
  3224--3229.

\bibitem{LHMNLC}
T.W. Stegink, C. {De Persis}, and A.J. {van der Schaft},
  \emph{Port-{H}amiltonian formulation of the gradient method applied to smart
  grids}, IFAC-PapersOnLine 48 (2015), pp. 13--18.

\bibitem{stegink2017unifying}
T.W. Stegink, C. {De Persis}, and A.J. {van der Schaft}, \emph{A unifying
  energy-based approach to stability of power grids with market dynamics}, IEEE
  Transactions on Automatic Control 62 (2017), pp. 2612--2622.

\bibitem{anderson1977}
P.M. Anderson and A.A. Fouad, \emph{Power System Control and Stability}, 1st
  ed., The Iowa State Univsersity Press, 1977.

\bibitem{powsysdynwiley}
J. Machowski, J.W. Bialek, and J.R. Bumby, \emph{Power System Dynamics:
  Stability and Control}, 2nd ed., John Wiley \& Sons, Ltd, 2008.

\bibitem{kundur}
P. Kundur, \emph{Power System Stability and Control}, Mc-Graw-Hill Engineering,
  1993.

\bibitem{cal-tab}
S.Y. Caliskan and P. Tabuada, \emph{Compositional transient stability analysis
  of multimachine power networks}, IEEE Transactions on Control of Network
  systems 1 (2014), pp. 4--14.

\bibitem{ortega2005transient}
R. Ortega, M. Galaz, A. Astolfi, Y. Sun, and T. Shen, \emph{Transient
  stabilization of multimachine power systems with nontrivial transfer
  conductances}, IEEE Transactions on Automatic Control 50 (2005), pp. 60--75.

\bibitem{AGC_ACC2014}
N. Li, L. Chen, C. Zhao, and S.H. Low, \emph{Connecting automatic generation
  control and economic dispatch from an optimization view}, in \emph{American
  Control Conference}. IEEE, 2014, pp. 735--740.

\bibitem{you2014reverse}
Y. Seungil and C. Lijun, \emph{Reverse and forward engineering of frequency
  control in power networks}, in \emph{Proc. of IEEE Conference on Decision and
  Control, Los Angeles, CA, USA}. 2014.

\bibitem{zhangpapaautomatica}
X. Zhang and A. Papachristodoulou, \emph{A real-time control framework for
  smart power networks: Design methodology and stability}, Automatica 58
  (2015), pp. 43--50.

\bibitem{zhao2015distributedAC}
C. Zhao, E. Mallada, and S.H. Low, \emph{Distributed generator and load-side
  secondary frequency control in power networks}, in \emph{49th Annual
  Conference on Information Sciences and Systems (CISS)}. IEEE, 2015, pp. 1--6.

\bibitem{pai1989energy}
M. Pai, \emph{Energy Function Analysis for Power System Stability}, Springer
  Science \& Business Media, 1989.

\bibitem{fouad1981transient}
A. Fouad and S. Stanton, \emph{Transient stability of a multi-machine power
  system part i: Investigation of system trajectories}, IEEE Transactions on
  Power Apparatus and Systems  (1981), pp. 3408--3416.

\bibitem{michel1983power}
A. Michel, A. Fouad, and V. Vittal, \emph{Power system transient stability
  using individual machine energy functions}, IEEE Transactions on Circuits and
  Systems 30 (1983), pp. 266--276.

\bibitem{trip2016internal}
S. Trip, M. B{\"u}rger, and C. {De Persis}, \emph{An internal model approach to
  (optimal) frequency regulation in power grids with time-varying voltages},
  Automatica 64 (2016), pp. 240--253.

\bibitem{de2018bregman}
C. De~Persis and N. Monshizadeh, \emph{Bregman storage functions for microgrid
  control}, IEEE Transactions on Automatic Control 63 (2018), pp. 53--68.

\bibitem{de2016lyapunov}
C. {De Persis}, N. Monshizadeh, J. Schiffer, and F. D{\"o}rfler, \emph{A
  {L}yapunov approach to control of microgrids with a network-preserved
  differential-algebraic model}, in \emph{IEEE Conference on Decision and
  Control}. 2016, pp. 2595--2600.

\bibitem{caliskan2015uses}
S.Y. Caliskan and P. Tabuada, \emph{Uses and abuses of the swing equation
  model}, in \emph{IEEE Conference on Decision and Control}. 2015, pp.
  6662--6667.

\bibitem{alv_meng_power_coupl_market}
F.L. Alvarado, J. Meng, C.L. DeMarco, and W.S. Mota, \emph{Stability analysis
  of interconnected power systems coupled with market dynamics}, IEEE
  Transactions on Power Systems 16 (2001), pp. 695--701.

\bibitem{stegink2016optimal}
T.W. Stegink, C. {De Persis}, and A.J. {van der Schaft}, \emph{Optimal power
  dispatch in networks of high-dimensional models of synchronous machines}, in
  \emph{55th IEEE Conference on Decision and Control (CDC)}. 2016, pp.
  4110--4115.

\bibitem{park1929two}
R.H. Park, \emph{Two-reaction theory of synchronous machines generalized method
  of analysis-part {I}}, IEEE Transactions of the American Institute of
  Electrical Engineers 48 (1929), pp. 716--727.

\bibitem{schiffer2016survey}
J. Schiffer, D. Zonetti, R. Ortega, A.M. Stankovi{\'c}, T. Sezi, and J. Raisch,
  \emph{A survey on modeling of microgrids-{F}rom fundamental physics to
  phasors and voltage sources}, Automatica 74 (2016), pp. 135--150.

\bibitem{sauerpai1998powersystem}
P.W. Sauer and M.A. Pai, \emph{Power system dynamics and stability},
  Prentice-Hall, 1998.

\bibitem{arjansteginkPerspectiveAnn2016}
A.J. {van der Schaft} and T.W. Stegink, \emph{Perspectives in modeling for
  control of power networks}, Annual Reviews in Control 41 (2016), pp.
  119--132.

\bibitem{maschke2000energy}
B. Maschke, R. Ortega, and A.J. Van Der~Schaft, \emph{Energy-based {L}yapunov
  functions for forced {H}amiltonian systems with dissipation}, IEEE
  Transactions on Automatic Control 45 (2000), pp. 1498--1502.

\bibitem{ahmed1982reduced}
S. Ahmed-Zaid, P.W. Sauer, M.A. Pai, and M.K. Sarioglu, \emph{Reduced order
  modeling of synchronous machines using singular perturbation}, IEEE
  Transactions on Circuits and Systems 29 (1982), pp. 782--786.

\bibitem{kokotovic1980singular}
P.V. Kokotovic, J.J. Allemong, J.R. Winkelman, and J.H. Chow, \emph{Singular
  perturbation and iterative separation of time scales}, Automatica 16 (1980),
  pp. 23--33.

\bibitem{swing-claudio}
M. Bürger, C. {De Persis}, and S. Trip, \emph{An internal model approach to
  (optimal) frequency regulation in power grids}, in \emph{Proceedings of the
  MTNS}, Groningen. 2014, pp. 577--583.

\bibitem{firstSPM}
A.R. Bergen and D.J. Hill, \emph{Structure preserving model for power system
  stability analysis}, IEEE Transaction on Power Apparatus and Systems PAS-100
  (1981), pp. 25--35.

\bibitem{zhangpapa}
X. Zhang and A. Papachristodoulou, \emph{A real-time control framework for
  smart power networks with star topology}, in \emph{American Control
  Conference}. IEEE, 2013, pp. 5062--5067.

\bibitem{DAPIsimpson2013synchronization}
J.W. Simpson-Porco, F. D{\"o}rfler, and F. Bullo, \emph{Synchronization and
  power sharing for droop-controlled inverters in islanded microgrids},
  Automatica 49 (2013), pp. 2603--2611.

\bibitem{zhang1achieving}
X. Zhang, N. Li, and A. Papachristodoulou, \emph{Achieving real-time economic
  dispatch in power networks via a saddle point design approach}, in
  \emph{Power \& Energy Society General Meeting}. IEEE, 2015, pp. 1--5.

\bibitem{bretas2003lyapunov}
N.G. Bretas and L.F. Alberto, \emph{{L}yapunov function for power systems with
  transfer conductances: extension of the invariance principle}, IEEE
  Transactions on Power Systems 18 (2003), pp. 769--777.

\bibitem{trip2017distributed}
S. Trip and C. De~Persis, \emph{Distributed optimal load frequency control with
  non-passive dynamics}, IEEE Transactions on Control of Network Systems
  (2017).

\end{thebibliography}

\end{document}